\documentclass[journal]{IEEEtran}

\usepackage{graphicx}
\usepackage[colorlinks,bookmarksopen,bookmarksnumbered,citecolor=blue,urlcolor=red]{hyperref}
\usepackage{amssymb}
\usepackage{amsbsy}
\usepackage{ulem}
\usepackage{upgreek}
\usepackage{algorithm}
\usepackage{algorithmicx}
\usepackage{algpseudocode}
\usepackage{threeparttable}
\usepackage[cmex10]{amsmath}
\usepackage{array}
\usepackage[tight,footnotesize]{subfigure}
\usepackage{caption}
\usepackage{url}
\usepackage{epstopdf}
\usepackage{verbatim}
\usepackage{color}
\usepackage{wrapfig}
\graphicspath{{figures/}}

\begin{document}

\newtheorem{thm}{Theorem}%[section]
\newtheorem{lma}{Lemma}%[theorem]
\newtheorem{defi}{Definition}
\newtheorem{proper}{Property}%[theorem]

\title{On New Approaches of Maximum Weighted Target Coverage and Sensor Connectivity: Hardness and Approximation}

\author{Ngoc-Tu~Nguyen,~
    Bing-Hong~Liu,~and
    Shih-Yuan~Wang
\IEEEcompsocitemizethanks{
\IEEEcompsocthanksitem
Corresponding author: B.H.~Liu.
\IEEEcompsocthanksitem
N.T.~Nguyen is with the Department of Computer Science and Engineering,
University of Minnesota, Twin Cities, Minneapolis, MN 55455 USA (e-mail:,
nguy3503@umn.edu).
\IEEEcompsocthanksitem B.H.~Liu and S.Y.~Wang are
with the Department of Electronic Engineering, National
Kaohsiung University of Science and Technology, 415, Chien Kung Rd., Kaohsiung
80778, Taiwan (e-mail: bhliu@nkust.edu.tw, ads7452@gmail.com).
}
}

\IEEEcompsoctitleabstractindextext{%
\begin{abstract}
In mobile wireless sensor networks (MWSNs), each sensor has the
ability not only to sense and transmit data but also to move to some
specific location. Because the movement of sensors consumes much more
power than that in sensing and communication, the
problem of scheduling mobile sensors to cover all targets and maintain network connectivity such that the total movement distance of
mobile sensors is minimized has received a great deal of attention.
However, in reality, due to a limited budget or numerous
targets, mobile sensors may be not enough to cover all targets or
form a connected network. Therefore, targets must be weighted by
their importance. The more important a target, the higher the weight
of the target. A more general problem for target coverage
and network connectivity, termed the Maximum Weighted Target
Coverage and Sensor Connectivity with Limited Mobile Sensors
(MWTCSCLMS) problem, is studied. In this paper, an approximation
algorithm, termed the weighted-maximum-coverage-based
algorithm (WMCBA), is proposed for the subproblem of the
MWTCSCLMS problem. Based on the WMCBA, the
Steiner-tree-based algorithm (STBA) is proposed for the MWTCSCLMS
problem. Simulation results demonstrate that the STBA provides
better performance than the other methods.
\end{abstract}
}

\maketitle

\IEEEdisplaynotcompsoctitleabstractindextext

\IEEEpeerreviewmaketitle

\section{Introduction}
Because of the rapid expansion of technology, well-developed
sensors accompany various sensing functions, such as
detecting surrounding temperatures, illuminations, and voices,
calculating and processing received information, and the communication
ability to transmit and receive data, which can be composed to form
a wireless sensor network
\cite{4359020,DAC:DAC2355,DAC:DAC2412,Wang2015330}. Recently, wireless sensor networks have been widely applied to
surveillance, security, and tracking
applications \cite{5290389,5705791,Somov2013217,6715654}. In these
applications, with the communication ability, sensors can
communicate with the data sink or other sensors to transmit sensed
data \cite{6522427}.

Due to the rapid development of sensor technology, in addition to
sensing ability and data transmission, sensors, also known as
mobile sensors, can have the ability to move to some locations. A
wireless sensor network that is composed of mobile sensors is also
known as a mobile wireless sensor network (MWSN). Because the
movement of sensors requires significantly higher power
consumption than that in sensing and communication \cite{5161266},
minimizing the total movement distance of mobile sensors becomes a more
important issue in MWSNs \cite{6846302}. In \cite{6846302}, the
Mobile Sensor Deployment (MSD) problem, which is the problem of
scheduling mobile sensors to cover all targets and maintain network
connectivity such that the total movement distance of mobile sensors
is minimized, is studied. For the MSD problem, algorithms based on
the clique partition and the Voronoi partition are proposed to find coverage
sensors to cover targets. In addition, the Euclidean minimum
spanning tree is used to span coverage sensors and the data sink,
and determine some points in the sensing field such that the network
composed of the sensors deployed on the points can form a connected
network. Finally, the Hungarian method is applied to schedule
adaptive mobile sensors to the generated points such that the total
movement distance is minimized.

Most research studies on target coverage when the number
of mobile sensors is assumed to be high enough such that a
connected network can always be formed to cover all targets.
However, in reality, due to a limited budget or numerous
targets, there may not be enough mobile sensors to cover all targets or
form a connected network. Therefore, targets must be weighted by
their importance. The more important a target, the higher the weight
of the target. This motivated us to study a more general and
practical problem for target coverage and network connectivity,
termed the Maximum Weighted Target Coverage and Sensor Connectivity
with Limited Mobile Sensors (MWTCSCLMS) problem. The MWTCSCLMS
problem is the problem of scheduling limited mobile sensors to
appropriate locations to cover targets and form a connected network
such that the total weight of the covered targets is maximized. The
highlights of the contribution in this paper are listed as follows:

\begin{itemize}
  \item A general problem for target coverage and network
  connectivity in MWSNs, termed the MWTCSCLMS problem, and its
  difficulty are introduced and discussed in this paper. In
  addition, when the transmission range is assumed to be large enough for any communication,
  a subproblem of the MWTCSCLMS problem, termed the Reduced MWTCSCLMS (RMWTCSCLMS)
  problem, and its difficulty are also introduced and discussed.
  \item An approximation algorithm, termed the weighted-maximum-coverage-based
algorithm (WMCBA), with an approximation ratio of $1-1/e$ is
proposed for the RMWTCSCLMS problem, where $e$ denotes the base of
the natural logarithm. In the WMCBA, all possible sets of targets
that can be covered by a mobile sensor located at any point in the
sensing field are considered. Then, a greedy method is used to
select suitable sets of targets to be covered by mobile sensors.
  \item Based on the WMCBA, the Steiner-tree-based
algorithm (STBA) is proposed for the MWTCSCLMS problem. In the STBA,
the Fermat points \cite{6023087} and a node-weighted Steiner tree
algorithm \cite{Klein1995104} are used to find a tree such that the
number of mobile sensors deployed by the tree structure to form a
connected network is minimized.
  \item Theoretical analyses of the WMCBA and the STBA are
  provided.
  \item Simulation results demonstrate that even if the number of
mobile sensors is high enough such that a connected network can
always be formed to cover all targets, the STBA requires a
significantly lower total movement distance than the best solution
proposed for the MSD problem \cite{6846302}. In addition, when the
mobile sensors may be not enough to cover all targets, the STBA
works better than the greedy method proposed in the simulation section
of this paper.
\end{itemize}

The remaining sections of this paper are organized as follows.
Related work is introduced in Section \ref{section:relaed work}.
In Section \ref{section:MWTCSCLMA and Difficulty}, illustrates the
MWTCSCLMS problem and the RMWTCSCLMS problem are illustrated. In addition, the
analyses of their difficulties are also provided. In Section
\ref{section:approximation_method}, the WMCBA is proposed for the
RMWTCSCLMS problem. In addition, the STBA is proposed in Section
\ref{section:method}. The performance of the STBA is evaluated in
Section \ref{section:Simulation}. The paper is concluded in Section
\ref{section:Conclusion}.

\section{Related Work}\label{section:relaed work}
The coverage problem is an important issue in a wireless sensor
network, in which each sensor has its own mission to monitor a
region through the sensor's sensing range. Different applications have various coverage
requirements
\cite{6314220,4703226,5719526,Mostafaei201551}. In \cite{6314220},
the area coverage problem is discussed in the way to deploy sensors
to form a wireless sensor network such that a particular area will
be fully covered and ensure the network connectivity. In
\cite{4703226}, the area coverage problem is studied to deploy
sensors to form a connected wireless sensor network even if
unpredicted obstacles exist in the sensing field. In \cite{5719526},
the problem of constructing a minimum size connected wireless sensor
network such that the critical grids in a sensing field are all
covered by sensors is addressed. In \cite{Mostafaei201551}, the
barrier coverage problem, the problem of deploying sensors to
construct a barrier such that invaders will be detected by at
least one sensor, is studied.

The target coverage problem is one of the coverage problems. In the
target coverage problem, targets are the points of interest (POI) in
the sensing field that are required to be covered and monitored by
sensors. In addition, the wireless sensor network composed of
sensors has to be connected such that the monitoring information
generated by the sensors can be reported to the data sink. When the sensors
are activated to monitor targets or transmit data, the sensors will
continuously consume energy. Therefore, the sensors will not be
able to monitor targets or transmit data if their energy is
exhausted. Because the energy of the sensors is often limited, many
studies have investigated extending the network lifetime to cover targets.
In \cite{6637016}, the problem of deploying sensors and scheduling
the sensors' activation time is studied such that all targets can be
covered and the network lifetime can be extended. In some cases, it
is hard for people to deploy sensors manually, and therefore,
random deployment \cite{1368897,Zhu2012619} can be used to construct
a wireless sensor network. Because random deployment cannot ascertain the sensors' locations before deployment, the problem of
scheduling sensors to be activated to form a wireless sensor network
and covering targets such that the network lifetime is extended
has received a great deal of attention \cite{6883345,4457911,6811184}.
In \cite{6883345}, a distributed algorithm is proposed to
alternatively activate sensors to form a minimal set cover for
covering all targets such that the network lifetime is maximized in
energy-harvesting wireless sensor networks. In \cite{4457911}, a
heuristic algorithm is proposed to schedule sensors into multiple
sets such that the sensors in each set can cover all targets and
form a connected network with the data sink. In addition, the sensor
sets are activated one-by-one such that the network lifetime can be
maximized. In \cite{6811184}, a polynomial-time constant-factor
approximation algorithm is proposed to schedule sensors to form a
connected network that can cover all targets and maximize the
network lifetime.

In MWSNs, when mobile sensors are randomly deployed in a sensing
field, mobile sensors can be used to improve the coverage quality
and the network connectivity in MWSNs. In \cite{Yang2010409}, a
survey on utilizing node mobility to extend the network lifetime is
discussed and provided. In \cite{4407689}, algorithms are proposed
to dispatch mobile sensors to designated locations such that the
area of interest can be $k$-covered. In \cite{6763013}, when mobile
sensors have different sensing ranges, algorithms based on the
multiplicatively weighted Voronoi diagram are proposed to find
coverage holes such that the coverage area can be improved. In
\cite{ElKorbi2014247}, an algorithm is proposed to relocate the minimum
number of redundant mobile sensors to maintain connectivity between
a region of interest and a center of interest in which a particular
event occurs, where mobile sensors are initially deployed in the
region of interest, and the center of interest is outside the region
of interest. In \cite{Zorbas20131039}, a distributed algorithm is
proposed to move mobile sensors to cover all targets and satisfy
the minimum allowed detection probability such that the network lifetime
is maximized.

\section{The Maximum Weighted Target Coverage and Sensor Connectivity with Limited Mobile Sensors Problem and Its
Difficulty}\label{section:MWTCSCLMA and Difficulty} The system model
used in this paper is illustrated in Section \ref{section:system
model}. Our problem, termed the Maximum Weighted Target Coverage and
Sensor Connectivity with Limited Mobile Sensors (MWTCSCLMS) problem,
is presented in Section \ref{sec:model}. Finally, the problem's difficulty is
analyzed in Section \ref{sec:hardness}.

\subsection{System Model}\label{section:system model}

In the MWSN, mobile sensors are responsible for sensing targets,
collecting sensed data, and reporting the data to a special node,
termed the data sink. The data sink can collect mobile sensors' location information and broadcast
deployment orders to mobile sensors \cite{6846302}.
For data collection, a mobile sensor $s$ can
sense and collect data from a target $t$ if $t$ is within $s$'s
sensing range, denoted by $R_s$. Hereafter, the target $t$ is said
to be covered if and only if $t$ is within at least one mobile
sensor's sensing range. Because some targets may be outside the sensing range of a mobile
sensor in the initial deployment \cite{Zhu2012619}, mobile
sensors must move to cover targets if necessary. Once targets are
covered or sensed by a mobile sensor $s$, the sensed data are generated
by $s$, and have to be reported to the data sink. In the MWSN, every
mobile sensor $s$ can transmit data to other mobile sensors within
its transmission range, denoted by $R_t$. The sensed data can then
be forwarded through sensors to the data sink by multi-hop protocols
\cite{6522427} if there exists a connected path from the node that generates
the sensed data to the data sink.
Take Fig. \ref{Fig:exam_MWSN}, for
example. In Fig. \ref{Fig:exam_MWSN}, it is clear that target $t_4$ is covered by mobile sensor $s_6$.
This is because $s_6$ is within the circle centered at $t_4$ with radius $R_s$, the distance between $s_6$ and $t_4$
is not greater than $R_s$. In addition, target $t_8$ can be covered by mobile sensor $s_{14}$ after the movement of $s_{14}$. It is also clear that the sensed data generated by $s_6$ can be forwarded to the data sink because the path from $s_6$ to the data sink is connected.

\begin{figure}
\center \subfigure{\includegraphics[width=8.5cm]{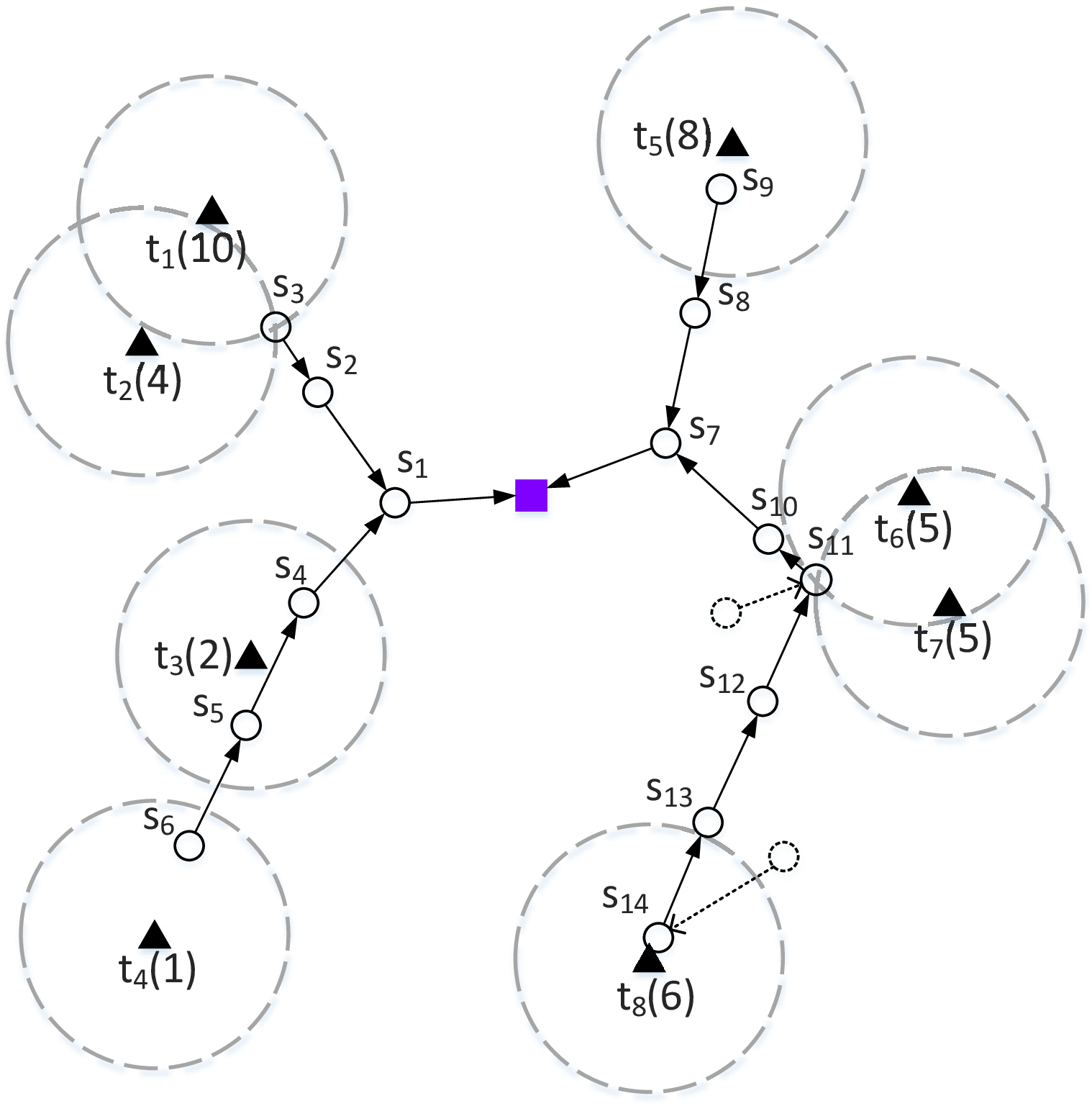}}\\
\subfigure{\includegraphics[width=6cm]{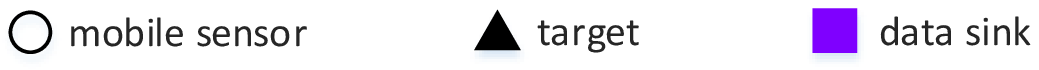}}
\caption{Example of the MWSN with $14$ mobile sensors and $8$
targets, where the number inside the parentheses indicates the corresponding
target's weight.} \label{Fig:exam_MWSN}
\end{figure}

In this paper, a set of $n$ mobile sensors $S=\left\{
{{s}_{1}},{{s}_{2}},\ldots,{{s}_{n}} \right\}$ is pre-deployed in a
sensing field. We assume that each mobile sensor in $S$ has the same
sensing range $R_s$ to sense targets. In addition, the data sink and
each mobile sensor have the same transmission range $R_t$ to
communicate with the other mobile sensors. While given a set of $m$
targets $T=\left\{ {{t}_{1}},{{t}_{2}},\ldots,{{t}_{m}} \right\}$
with known locations in the field, mobile sensors can be scheduled
to move in any direction and stop anywhere \cite{5161266} to cover targets or
connect with the data sink and the other mobile sensors. In reality, all
of the targets in the field may not be covered due to the limited mobile
sensors. Targets in the sensing field, therefore, must be weighted by
their importance; that is, the more important a target, the higher the weight of the target. Hereafter, the
weight of target $t$ is denoted by $t.\omega$.

\subsection{The MWTCSCLMS Problem} \label{sec:model}
In this paper, we study scheduling limited mobile sensors to
appropriate locations to cover targets and form a connected network
such that the total weight of the covered targets is maximized, termed
the Maximum Weighted Target Coverage and Sensor Connectivity with
Limited Mobile Sensors (MWTCSCLMS) problem. While given an MWSN with
a data sink, a set of deployed mobile sensors $S = \{s_1, s_2, \ldots,s_n\}$, and
a set of targets $T = \{t_1, t_2, \ldots,t_m\}$, the MWTCSCLMS
problem can be formally illustrated as follows:

\textbf{INSTANCE}: Given $R_s$, $R_t$, a data sink $sink$, a set of deployed mobile
sensors $S = \{s_1, s_2, \ldots,s_n\}$, and a set of targets $T = \{t_1,
t_2, \ldots,t_m\}$, where each sensor $s \in S$ has its own
position, and each target $t \in T$ has its weight $t.\omega$.

\textbf{QUESTION}: Does there exist a schedule of mobile sensors in
an MWSN for target coverage and network connectivity such that the
total weight of the covered targets is maximized?

The MWTCSCLMS problem can be viewed under two issues,
target coverage and network connectivity. For target coverage, we
can schedule mobile sensors to maximize the total weight of the covered
targets. For network connectivity, the remaining mobile sensors can be
scheduled to form a connected network such that the data generated
from sensing targets can be forwarded to the data sink. When given an
MWSN as in Fig. \ref{Fig:exam_MWSN}, it is clear that the data sink
and $14$ mobile sensors form a connected network. In addition, because all
targets can be covered by the connected
network, the total weight of the covered targets is $10$ $+$ $4$ $+$ $2$ $+$ $1$ $+$ $8$ $+$ $5$ $+$ $5$ $+$ $6$ $=$
$41$.

\subsection{Difficulty of the MWTCSCLMS Problem}\label{sec:hardness}
In this subsection, a special case of the MWTCSCLMS problem, termed
the Reduced MWTCSCLMS (RMWTCSCLMS) problem, is presented to show the
difficulty of the MWTCSCLMS problem. In the RMWTCSCLMS problem, when
$R_s$, a data sink $sink$, a set of deployed mobile sensors $S = \{s_1, s_2,
\ldots,s_n\}$, and a set of targets $T = \{t_1, t_2, \ldots,t_m\}$ are
given, and $R_t$ is set to be large enough such that any two mobile sensors
(or any one mobile sensor and the data sink) can communicate with each other, the RMWTCSCLMS problem is
scheduling mobile sensors in an MWSN for target coverage and network
connectivity such that the total weight of the covered targets is
maximized. We then show that the RMWTCSCLMS problem is NP-hard in
Lemma \ref{thm:RMWTCSCLMS_hardness}. By Lemma
\ref{thm:RMWTCSCLMS_hardness}, the difficulty of the MWTCSCLMS
problem is then concluded in Theorem \ref{thm:MWTCSCLMS_hardness}.

\begin{lma}\label{thm:RMWTCSCLMS_hardness}
The RMWTCSCLMS problem is NP-hard.
\end{lma}

\begin{proof}
Here, the Target COVerage (TCOV) \cite{6846302} problem is used to show that
the RMWTCSCLMS problem is NP-hard. While we are given a set of
deployed mobile sensors $S' = \{s_1, s_2, \ldots,s_{n'}\}$ each having
sensing range $R'_s$ and its own position, and a set of targets $T' =
\{t_1, t_2, \ldots,t_{m'}\}$, the TCOV problem is scheduling mobile
sensors in an MWSN to cover all targets such that the total movement distance
of the mobile sensors is minimized. Clearly, in the RMWTCSCLMS problem,
when $R_s$ $=$ $R'_s$, $R_t$ $=$ $\infty$, $S$ $=$ $S'$, $T$ $=$ $T'$, and $t.\omega$
$=$ $1$ for each $t \in T$, the TCOV problem is also an RMWTCSCLMS
problem. Therefore, we have that the TCOV problem is a subproblem of
the RMWTCSCLMS problem. Because the TCOV problem is NP-hard \cite{8356712},
the RMWTCSCLMS problem is thus NP-hard, which completes the proof.
\end{proof}

\begin{thm}\label{thm:MWTCSCLMS_hardness}
The MWTCSCLMS problem is NP-complete.
\end{thm}

\begin{proof}
Because the MWTCSCLMS problem clearly belongs to the NP class, it
suffices to show that the MWTCSCLMS problem is NP-hard. Because the RMWTCSCLMS problem, which is NP-hard by Lemma
\ref{thm:RMWTCSCLMS_hardness}, is a subproblem of the MWTCSCLMS
problem, the MWTCSCLMS problem is NP-hard, which completes the
proof.
\end{proof}

\section{Approximation Algorithm for a Special Case of the MWTCSCLMS Problem}\label{section:approximation_method}
In the section, we analyze a special case of the MWTCSCLMS problem,
that is, the RMWTCSCLMS problem, and present an approximation
algorithm for the problem accordingly. In the RMWTCSCLMS problem,
because $R_t$ is large enough such that any two mobile sensors (or
any one mobile sensor and the data sink) can communicate with each
other, the main task is to schedule limited mobile sensors to cover the
targets with the maximum total weight. Therefore, how to schedule
limited mobile sensors to cover which targets is important in the
RMWTCSCLMS problem. It is clear that if one mobile sensor can
exactly cover one target, the collection of possible sets of targets
covered by the mobile sensor is $\{\{t_1\}, \{t_2\},
\ldots,\{t_m\}\}$, and the cardinality of the collection is equal to
${\left(
\begin{matrix} m \\ 1  \\ \end{matrix} \right)}$. If one mobile
sensor can exactly cover $k$ targets, the cardinality of the
collection of possible sets of targets covered by the mobile sensor
is equal to ${\left(
\begin{matrix} m \\ k \\ \end{matrix} \right)}$. Because a mobile
sensor can move to cover $0$ or $k$ $(1 \leq k \leq m)$ targets, the
number of possible sets of targets covered by one mobile sensor is
therefore
$1$ $+$ $\sum\limits_{k=1}^{m}{\left(\begin{matrix} m \\
k  \\ \end{matrix} \right)}$ $=$ $2^m$. To solve the RMWTCSCLMS
problem, a brute-force algorithm can be used to check all possible
sets of targets such that the total weight of the targets covered by
the mobile sensors is maximized; however, the time complexity of the
brute-force algorithm is $O(2^{nm})$ because it has to check
$2^{nm}$ cases for $n$ mobile sensors. To overcome the challenge, an
approximation algorithm, termed the weighted-maximum-coverage-based
algorithm (WMCBA), which takes $O(m^3)$ time, is proposed for the
RMWTCSCLMS problem in Section \ref{section:wmcba}. In addition, the
theoretical analysis of the WMCBA is provided in Section
\ref{section:wmcba_analysis}.

\subsection{The WMCBA}\label{section:wmcba}
In the WMCBA, the idea is to transform any instance of the
RMWTCSCLMS problem into an instance of the Weighted Maximum Coverage
(WMC) problem. Then, an existing algorithm is used to find the
solution $SOL$ for the instance of the WMC problem. Finally, the
solution for the instance of the RMWTCSCLMS problem can thus be
obtained with $SOL$. In the WMC problem, while given an universal set
$U = \{u_1, u_2, \ldots, u_q\}$ with every element $u_i$ in $U$
having a weight $u_i.\tau$, a collection of sets of elements in $U$
$C = \{C_1, C_2, \ldots, C_r\}$, and a number $k$, the WMC problem
is to find a collection $C' \subseteq C$ such that $|C'| \leq k$ and
the total weight of $u_i$ for all $u_i \in \bigcup_{C_j \in
C'}{C_j}$ is maximized, where $|C'|$ denotes the cardinality of
$C'$. For example, while given an universal set $U$ $=$ $\{u_1, u_2, u_3, u_4, u_5, u_6\}$ with $u_i.\tau = 1$ ($1 \leq i \leq 6$), $C$ $=$ $\{\{u_1, u_2, u_3\}$, $\{u_2, u_4, u_6\}$, $\{u_4, u_5, u_6\}\}$, and $k = 2$, it is easy to verify that $C'$ $=$ $\{\{u_1, u_2, u_3\}$, $\{u_4, u_5, u_6\}\}$ has $|C'| \leq k$, and has maximal total weight $6$.

In the WMCBA, while given an instance of the
RMWTCSCLMS problem, including $R_s$, $R_t$, $sink$, $S$, and $T$, because
$R_t$ is large enough such that any two mobile sensors (or any one
mobile sensor and the data sink) can communicate with each other,
the network formed by the data sink and the mobile sensors must be connected. In addition,
the targets in $T$ can be treated as the elements in $U$ in the WMC
problem, the cardinality of $S$ can be treated as the number $k$ in
the WMC problem, and the set of targets covered by a mobile sensor
located at some position can be treated as some set in $C$. It is
clear that when we have a solution $SOL$ to the transformed instance
of the WMC problem, the solution for the original instance of the
RMWTCSCLMS problem can thus be obtained accordingly. From the
transformation, it is clear that how to find all possible sets of
targets that can be covered by mobile sensors with lower time
complexity and how to solve the WMC problem are critical issues in
the WMCBA.

Because each mobile sensor has sensing range $R_s$, a target $t_i$
is covered by a mobile sensor $s$ only if the distance between the
mobile sensor and the target is not greater than $R_s$. Let $O_{t_i}$ denote a circle
centered at $t_i$ with radius $R_s$. This also implies that the
mobile sensor $s$ is within the area enclosed by $O_{t_i}$. For two
targets $t_i$ and $t_j$, if $t_i$ and $t_j$ can be covered by a
mobile sensor $s$, it is clear that $s$ must be within the area
intersected by the circles $O_{t_i}$ and $O_{t_j}$. Therefore, when
we have a set of targets $P$ in which each target ${{t}_{i}} \in P$
can be covered by a mobile sensor $s$, $s$ must be within the
area intersected by the circles $O_{t_i}$ for all ${{t}_{i}} \in P$.
We know that when two circles centered at distinct positions with
radii $R_s$ intersect, at most two intersection points
exist and are located in the boundary of the intersection area. When
an area $A$ is generated by the intersection of the circles
$O_{t_i}$ for all ${{t}_{i}} \in P$, at least one intersection point
is generated and located in the boundary of the $A$. That is, at
least one intersection point can be selected to be the location of
the sensor $s$ such that $t_i$ can be covered by $s$ for all
${{t}_{i}} \in P$. Let $P_{point}$ denote the set of targets that
can be covered by a mobile sensor located at $point$. The collection
of all possible sets of targets $C_T$ that can be covered by mobile
sensors is constructed by the union of $\{\{{t}_{i}\}\}$ for all
${{t}_{i}} \in T$ and $\{P_{p^1_{t_i,t_j}},P_{p^2_{t_i,t_j}}\}$ for
any ${{t}_{i}},{{t}_{j}} \in T$, where $p^1_{t_i,t_j}$ and
$p^2_{t_i,t_j}$ denote the two intersection points intersected by
circles $O_{t_i}$ and $O_{t_j}$. Lemma \ref{lma:all_possible_set}
shows that all possible sets of targets that can be covered by
mobile sensors are included in $C_T$.

\begin{lma}\label{lma:all_possible_set}
For any point $p$ in the sensing field, the set of targets $P_{p}$ that can be covered by a mobile sensor located at $p$ must be included in $C_T$.
\end{lma}

\begin{proof}
Because $C_T$ contains $\{{t}_{i}\}$ for all ${{t}_{i}} \in T$, the case for a mobile sensor that
exactly covers a target is fully considered. Therefore, it suffices to show that the set of two or more targets
that can be covered by a mobile sensor located at $p$ must be included in $C_T$.
Assume that a set $P_{p'}$ $=$ $\{t_1, t_2, \ldots, t_{m'}\}$ whose targets can be covered by a mobile sensor located at $p'$ exists but is not included in $C_T$. This implies that the distance between $p'$ and $t_i$ is not greater than $R_s$ for all $t_i$ $\in$ $P_{p'}$. This also implies that $p'$ is within the area $A$ intersected by the circles centered at $t_i$ with radii $R_s$ for all $t_i$ $\in$ $P_{p'}$. Because $A$ is constructed by the intersection of the circles centered at $t_i$ for all $t_i$ $\in$ $P_{p'}$, there must exist at least one intersection point $p''$ in the boundary of $A$. This implies that the distance between
$p''$ and $t_i$ is not greater than $R_s$ for all $t_i$ $\in$ $P_{p'}$. This also implies that $P_{p''}$ $=$ $\{t_1, t_2, \ldots, t_{m'}\}$ $\in$ $C_T$ because $p''$ is an intersection point of circles. We have that $P_{p'}$ $=$ $P_{p''}$ $\in$ $C_T$, which constitutes a contradiction, and thus, completes the proof.
\end{proof}

In the WMCBA, how to solve the WMC problem is another critical
issue. Because the WMC problem is NP-hard \cite{Nemhauser1978}, a greedy algorithm
with an approximation ratio $1-1/e$ \cite{Nemhauser1978} is applied to the WMC
problem, where $e$ denotes the base of the natural logarithm. In the
greedy algorithm, the set with the maximum weight of uncovered elements
is selected in each iteration. The process is repeated until all
elements are covered or $k$ sets are selected.

While given an instance of the RMWTCSCLMS problem, including $R_s$, $R_t$, $sink$, $S = \{s_1, s_2,
\ldots,s_n\}$, and $T = \{t_1, t_2, \ldots,t_m\}$,
the WMCBA contains three steps that are illustrated in detail as follows:

\textbf{1) Construction of $U$, $C$, and $k$:} Let $U$ be the set of
nodes $u_i$ for each $t_i \in T$, where $u_i.\tau$ is set to
$t_i.\omega$ for each $u_i \in U$. Let $C$ be the union of
$\{\{u_{{i}}\}\}$ for all ${{t}_{i}} \in T$ and
$\{C_{p^1_{t_i,t_j}},C_{p^2_{t_i,t_j}}\}$ for any
${{t}_{i}},{{t}_{j}} \in T$, where $C_{p^1_{t_i,t_j}}$ (or
$C_{p^2_{t_i,t_j}}$) is the set of $u_x$ for each $t_x \in
P_{p^1_{t_i,t_j}}$ (or $P_{p^2_{t_i,t_j}}$). The $k$ is set to $n$.

\textbf{2) Establishment of Collection $C'$:} We apply the greedy algorithm \cite{Nemhauser1978} to find a collection $C' \subseteq C$.

\textbf{3) Scheduling of Mobile Sensors:}
Let $L$ $=$ $\{p_{t_z}\}$ $\cup$ $\{p^1_{t_x,t_y},p^2_{t_x,t_y}\}$ for all $C_{{t}_{z}}$, $C_{p^1_{t_x,t_y}}$, $C_{p^2_{t_x,t_y}}$ $\in$ $C'$, where $p_{t_z}$ denotes the location point of $t_z$. Mobile sensors in $S$ are scheduled to each
points in $L$.

Take the MWSN in Fig. \ref{Fig:exam_MWSN}, for example, where $R_t$ is assumed to be $\infty$.
In the construction of $U$, $C$, and $k$, $U$ is set to the union of
$\{u_{{i}}\}$ for $1 \leq i \leq 8$; $u_i.\tau$ is set to
$t_i.\omega$ for $1 \leq i \leq 8$; $C$ is set to the union of
$\{\{u_{{i}}\}\}$ for $1 \leq i \leq 8$ and $\{ \{u_1, u_2\}, \{u_6, u_7\} \}$; and $k$ is set to $14$.
When the greedy algorithm \cite{Nemhauser1978} is applied, $C'$ $=$ $\{\{u_1, u_2\}$, $\{u_6, u_7\}$, $\{u_3\}$, $\{u_4\}$, $\{u_5\}$, $\{u_8\}\}$ can be obtained. Then, mobile devices can be scheduled to $p^1_{t_1,t_2}$, $p^1_{t_6,t_7}$, $p_{t_3}$, $p_{t_4}$, $p_{t_5}$, and $p_{t_8}$.

To minimize the total movement distance of mobile sensors, the
Hungarian method \cite{NAV:NAV3800020109} is applied for assigning mobile sensors to
the points in $L$. The Hungarian method can be used to find an
optimal solution in polynomial time for the assignment problem. In
the assignment problem, when a set of agents $A$ and a set of tasks
$H$ are given and have the same cardinality, each agent $a \in A$
can be assigned to perform any task $h \in H$ with cost $\psi(a,h)$.
The assignment problem is to assign exactly one agent $a \in A$ to
each task and assign exactly one task $h \in H$ to each agent such
that the total cost of the assignment is minimized. It is clear that
a mobile sensor $s_i \in S$ can be regarded as an agent $a_i \in A$;
a point $p_j$ $\in$ $L$ can be regarded as a task $h_j \in H$; and
the distance required by $s_i$ to move to or cover $p_j$ can be
regarded as the cost $\psi(a_i,h_j)$. Therefore, the problem of
assigning mobile sensors in $S$ to the points in $L$ can be
transferred into the assignment problem and can be solved by the
Hungarian method \cite{NAV:NAV3800020109} if $S$ and $L$ have the same cardinality.
However, $|S|$ and $|L|$ are not always the same. Because less than
or equal to $|S|$ sets (or location points) are selected in the WMCBA, $|S|$
is greater than or equal to $|L|$. For this reason, we can obtain $L'$ by adding
some dummy points into $L$ such that $|S|$ $=$ $|L'|$, where the dummy points' corresponding costs $\psi$ are set to $0$.
Therefore, when $S$ and $L$ are given, the cost matrix generated by $S$ and $L$ for the input
of the Hungarian method is an $n$ $\times$ $n$ matrix and is shown as follows:
\[{{[\psi(s_i,p_j)]}_{n\times n}}=\left(\begin{matrix}
   {\psi(s_1,p_1)} & \ldots  & {\psi(s_1,p_\ell)} & 0 & \cdots  & 0 \\
   \vdots  & \ddots  & \vdots & \vdots  & \ddots  & \vdots \\
   {\psi(s_n,p_1)} & \cdots  & {\psi(s_n,p_\ell)} & 0 & \cdots  & 0 \\
\end{matrix} \right),\]
where $n$ $=$ $|S|$; $\ell$ $=$ $|L|$; $\psi(s_i,p_j)$ $=$ $\zeta (s_i,p_j)$ if $p_j$ is some $p^1_{t_x,t_y}$ or $p^2_{t_x,t_y}$ in $L$; $\psi(s_i,p_j)$ $=$ $\zeta (s_i,p_j)$ $-$ $R_s$ if $p_j$ is some $p_{t_z}$ in $L$ and $\zeta (s_i,p_j)$ $\geq$ $R_s$;
$\psi(s_i,p_j)$ $=$ $0$ for other cases; and $\zeta (s_i,p_j)$ denotes the distance
between $s_i$ and $p_j$. When the cost matrix is determined, the optimal assignment can be obtained with the Hungarian method \cite{NAV:NAV3800020109}.

\subsection{Theoretical Analysis of the WMCBA}\label{section:wmcba_analysis}
In the following, the analysis of the time complexity of the WMCBA
is given in Theorem \ref{thm:WMCBA_time_complexity}. In addition,
Lemma \ref{lma:strict-reduction} shows that there exists a strict
reduction from the RMWTCSCLMS problem to the WMC problem with the
WMCBA. Theorem \ref{thm:approximation_ratio} provides the
approximation ratio of the WMCBA with the help of Lemma
\ref{lma:strict-reduction}.

\begin{thm}\label{thm:WMCBA_time_complexity}
The time complexity of the WMCBA is bounded in $O(m^3)$, where $m$
is the number of targets.
\end{thm}

\begin{proof}
Because all targets can be covered if the number of mobile sensors
is greater than or equal to the number of targets, that is, $n \geq
m$, we discuss only the case with $n < m$ in the following. In the
construction of $U$, $C$, and $k$, because $T$ has $m$ elements, it
requires $O(m)$ time to construct $U$. Because there are at most two
intersection points for any two distinct circles, at most $2 \times
{\left(
\begin{matrix} m \\ 2 \\ \end{matrix} \right)}$ intersection points are generated.
Therefore, we have that there are at most $m$ + $2 \times {\left(
\begin{matrix} m \\ 2 \\ \end{matrix} \right)}$ $=$ $m^2$ elements in $C$.
In addition, because it needs at most $O(m)$ time to check if any targets are within a circle, it requires
$O(m \times m^2) = O(m^3)$ time to construct $C$. Therefore, the construction
of $U$, $C$, and $k$ requires $O(m^3)$ time because the setting of
$k$ requires only constant time. In the establishment of collection
$C'$, because the greedy algorithm \cite{Nemhauser1978} is applied to iteratively
select a set with the maximum weight of the uncovered elements in $C$ until
$k$ sets are selected or all elements in $U$ are covered, it
requires $O(k \times m^2)$ or $O(m \times m^2)$ time to construct $C'$
because there are at most $m^2$ elements in $C$ and at most $m$
iterations are required in the greedy algorithm. Because $k = n <
m$, it requires $O(m \times m^2)$ $=$ $O(m^3)$ time for the
construction of $C'$. Because at most $k$ mobile sensors are
scheduled and $k$ $<$ $m$, at most an $m$ $\times$ $m$ cost matrix is required for the
Hungarian method. By \cite{NAV:NAV3800020109}, we have that it requires $O(m^3)$ time for the Hungarian method, which implies that it requires $O(m^3)$ time for the
scheduling of mobile sensors. Therefore, the WMCBA requires $O(m^3)$
$+$ $O(m^3)$ + $O(m^3)$ $=$ $O(m^3)$ time, which completes the proof.
\end{proof}

\begin{lma}\label{lma:strict-reduction}
There exists a strict reduction from the RMWTCSCLMS problem to the
WMC problem by the WMCBA.
\end{lma}

\begin{proof}
Let $\Pi_1$ and $\Pi_2$ be the RMWTCSCLMS problem and the WMC
problem, respectively. While given any instance $I$ of $\Pi_1$,
including $R_s$, $R_t$, $S$, and $T$, the WMCBA can transform $I$
into an instance of $\Pi_2$, termed $f(I)$, including $U$, $C$, and
$k$, where $f$ denotes the function that works as the step $1$ of
the WMCBA. Let $g$ be the function that works as the step $3$ of the
WMCBA and can transform any feasible solution $S$ of $f(I)$ into a
feasible solution $g(S)$ of $I$. By Theorem
\ref{thm:WMCBA_time_complexity}, we have that $f$ and $g$ are
polynomial time computable functions because the WMCBA can be
executed in polynomial time. Therefore, it suffices to show that C1)
the optimal solution of $f(I)$ can lead to an optimal solution of
$I$, and C2) any feasible solution of $f(I)$ can lead to a feasible
solution of $I$ with a better or equivalent performance ratio
\cite{Chen2011}. The proof of C2 is omitted here due to the
similarity of the proof of C1.

For C1, when an optimal solution $S^2_{OPT}$ of $f(I)$ is given,
assume that $g(S^2_{OPT})$ is not an optimal solution of $I$, that
is, there exists an optimal solution $S^1_{OPT}$ of $I$ such that
$c_1(S^1_{OPT}) > c_1(g(S^2_{OPT}))$, where $c_1(S)$ denotes a cost
function and produces the total weight of the covered targets for each
feasible solution $S$. By Lemma \ref{lma:all_possible_set} and the
construction of $C$, we have that all possible sets of targets that
can be covered by mobile sensors are considered and included in $C$;
that is, if a set of targets that can be covered by a mobile sensor
exists, the corresponding set also exists in $C$. Therefore, any set
of targets that can be covered by mobile sensors in $S^1_{OPT}$ has
a corresponding set in $C$. Then we can construct a feasible
solution $S^2_{S^1_{OPT}}$ to $f(I)$ by selecting the corresponding
set in $C$ for each set of targets covered by mobile sensors in
$S^1_{OPT}$. Because $u_i.\tau$ is equal to $t_i.\omega$ for each
$u_i \in U$, the total weight of the covered elements in
$S^2_{S^1_{OPT}}$ is equal to the total weight of the covered targets in
$S^1_{OPT}$. Let $c_2 (S)$ denote a cost function and produce the
total weight of the covered element for each feasible solution $S$. We
thus have that $c_2 (S^2_{S^1_{OPT}})$ $=$ $c_1(S^1_{OPT})$. In a
similar way, we also have that $c_1(g(S^2_{OPT}))$ $\geq$ $c_2
(S^2_{OPT})$. Because $c_1(S^1_{OPT}) > c_1(g(S^2_{OPT}))$, we have
that $c_2 (S^2_{S^1_{OPT}})$ $=$ $c_1(S^1_{OPT})$ $>$
$c_1(g(S^2_{OPT}))$ $\geq$ $c_2 (S^2_{OPT})$, which implies that
$c_2 (S^2_{S^1_{OPT}})$ $>$ $c_2 (S^2_{OPT})$. This implies that
$S^2_{OPT}$ is not an optimal solution of $f(I)$, which constitutes
a contradiction, and thus, completes the proof.
\end{proof}

\begin{thm}\label{thm:approximation_ratio}
The WMCBA achieves an approximation ratio of $1-1/e$ for the
RMWTCSCLMS problem, where $e$ denotes the base of the natural
logarithm.
\end{thm}

\begin{proof}
By \cite{Chen2011}, if there exists a strict reduction from
$\Pi_1$ to $\Pi_2$, in which $\Pi_1$ and $\Pi_2$ represent two
optimization problems, any existing $\rho$-approximation algorithm
of $\Pi_2$ can lead to a $\rho$-approximation algorithm of $\Pi_1$.
By Lemma \ref{lma:strict-reduction}, it implies that any existing
$\rho$-approximation algorithm of the WMC problem can lead to a
$\rho$-approximation algorithm of the RMWTCSCLMS problem by the
WMCBA. Because a greedy algorithm \cite{Nemhauser1978} with approximation ratio
$1-1/e$ for the WMC problem is applied in the WMCBA, the WMCBA has
an approximation ratio $1-1/e$ to the RMWTCSCLMS problem, which
completes the proof.
\end{proof}

\section{Algorithm for the MWTCSCLMS Problem}\label{section:method}
Because the MWTCSCLMS problem is to schedule limited mobile sensors
to appropriate locations to cover targets and form a connected
network, the proposed algorithm, termed the Steiner-tree-based
algorithm (STBA), is designed to determine appropriate locations
first, termed the potential points, and then, move mobile sensors
to the potential points for target coverage and network
connectivity. Hereafter, a set of points is said to be connected or
form a connected network if the network with the data sink and the
mobile sensors located at the points is connected. In addition, a
target is said to be covered by a point if a mobile sensor located
at the point can cover the target; and a set of targets is said to
be covered by a point if each target in the set is covered by the
point.

Because only limited mobile sensors can be used to cover targets and
form a connected network, the idea of the STBA is to iteratively add
potential points to cover some adaptive targets and form a network
connected with the data sink, until there are not enough mobile sensors or all targets are covered. Because a potential
point can cover one or more targets, how to determine the positions
of potential points in a sensing field for target coverage and
network connectivity are critical issues in the MWTCSCLMS problem.
By Lemma \ref{lma:all_possible_set}, because all possible sets of
targets that can be covered by any point in a sensing field are
considered in the construction of $C$ of the WMCBA, the location
points $p_{t_z}$ for all $t_z \in T$ and the intersection points
$p^1_{t_i,t_j}$ and $p^2_{t_i,t_j}$ for any $t_i,t_j \in T$ are
considered to be the reference points that can be used to be the
guides for generating potential points. Let $X_1$ be the set of
location points ${p}_{t_z}$ for all ${{t}_{z}}$ $\in$ $T$; and let
$X_2$ be the set of intersection points ${p}^1_{t_i,t_j}$ and
${p}^2_{t_i,t_j}$ for any $t_i,t_j$ $\in$ $T$. Let $X$ $=$ $X_1$
$\cup$ $X_2$. The points $p \in X$ can then be used as the guides
for generating potential points. In addition, $P_{{p}_{t_z}}$ is set
to $\{{t}_{z}\}$ for each ${p}_{t_z} \in X_1$; and
$P_{p^1_{t_i,t_j}}$ and $P_{p^2_{t_i,t_j}}$ for any
${p}^1_{t_i,t_j}$, ${p}^2_{t_i,t_j}$ $\in$ $X_2$ are set to the sets
of targets covered by points ${p}^1_{t_i,t_j}$ and
${p}^2_{t_i,t_j}$, respectively. The $P_{{p}_{t_z}}$,
$P_{p^1_{t_i,t_j}}$, and $P_{p^2_{t_i,t_j}}$ are similar to the
elements in the construction of $C$ of the WMCBA, and are used to
show which targets can be covered when a mobile sensor is located at
${p}_{t_z}$, $p^1_{t_i,t_j}$, or $p^2_{t_i,t_j}$. Therefore, when a
set of targets $P_{{p}_{t_z}}$, $P_{p^1_{t_i,t_j}}$, or
$P_{p^2_{t_i,t_j}}$ is selected to be covered, the corresponding
reference point ${p}_{t_z}$, $p^1_{t_i,t_j}$, or $p^2_{t_i,t_j}$ can
be regarded as a guide to generate potential points to cover the targets
in $P_{{p}_{t_z}}$, $P_{p^1_{t_i,t_j}}$, or $P_{p^2_{t_i,t_j}}$ and
form a connected network. Take the MWSN in Fig. \ref{Fig:exam_MWSN}, for example.
It is clear that $X_1$ is the set of
location points ${p}_{t_i}$ for $1 \leq i \leq 8$; and $X_2$ is $\{{p}^1_{t_1,t_2}, {p}^2_{t_1,t_2}, {p}^1_{t_6,t_7}, {p}^2_{t_6,t_7}\}$. The $X$ is set to the union of $X_1$ and $X_2$, and the reference points in $X$ are shown in Fig. \ref{Fig:exam_reference_points}.

\begin{figure}
\center \subfigure{\includegraphics[width=8.5cm]{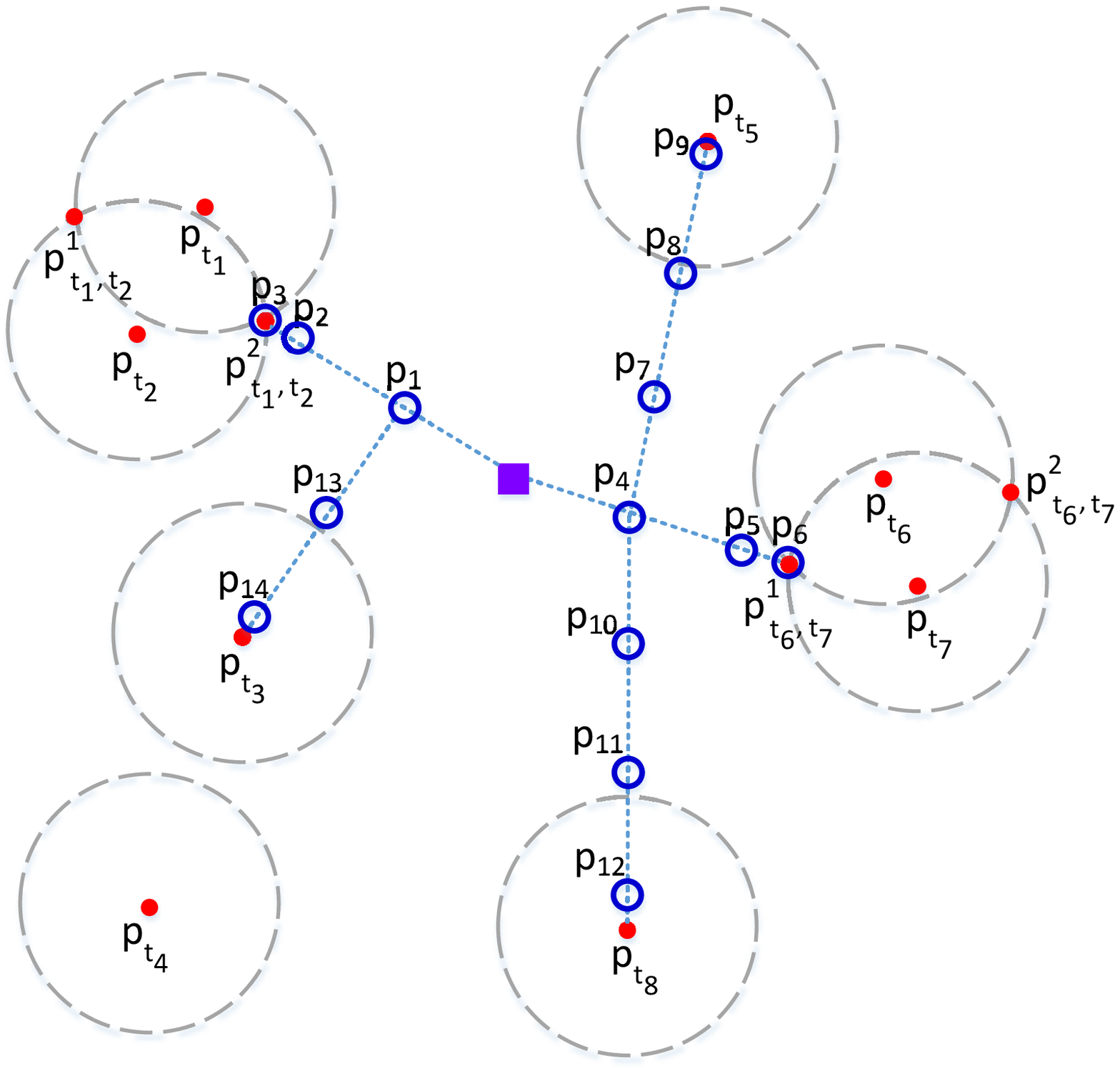}}\\
\subfigure{\includegraphics[width=6cm]{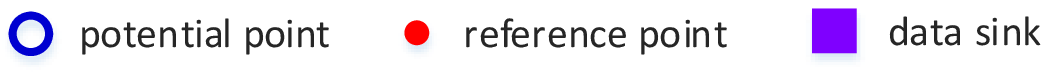}}
\caption{Example of reference points and potential points obtained by the STBA from the MWSN shown in Fig. \ref{Fig:exam_MWSN}.
} \label{Fig:exam_reference_points}
\end{figure}

Before deciding which set of targets $P_p$ for any reference point
$p \in X$ to be covered in each iteration in the STBA, we have to
know how much cost to pay for covering $P_p$, that is, how many
additional potential points are required. Let $N$ be a set of potential points and the data sink,
which can form a network connected with the data sink. Also let
$\eta_N(p)$ denote the minimum distance between point $p$ and each
point in $N$. When a set of targets $P_p$, in which $p$ is some
reference point $p^1_{t_i,t_j}$ or $p^2_{t_i,t_j}$ in $X$, is
considered to be covered and form a connected network with $N$, it
is clear that at least $\left\lceil \frac{\eta_N(p)}{{{R}_{t}}}
\right\rceil$ additional potential points are required to cover
$P_p$ and form a connected network with $N$. That is, $\left\lceil
\frac{\eta_N(p)}{{{R}_{t}}} \right\rceil$ potential points each can
be generated on the straight line from $p'$ to $p$ every distance
$R_t$, not including $p'$, until $p$ is reached, where $p'$ is a point in $N$ that has the minimum
distance to $p$. In addition, when a set of targets $P_p$, in which
$p$ is some reference point ${p}_{t_z}$ in $X$, is considered to be
covered and forms a connected network with $N$, at least $\left\lceil
\frac{\eta_N(p)-{{R}_{s}}}{{{R}_{t}}} \right\rceil$ potential points
are required because any point within the circle centered at $p_{t_z}$
with radius $R_s$ can cover target $t_z$.
Let $\phi_N (p)$ denote
the number of additional potential points required to cover
$P_p$ and form a connected network with $N$. The $\phi_N (p)$ can
then be defined in Eq. \ref{equa_phi}:
\begin{equation}\label{equa_phi}
\phi_N (p)=\left\{^{\left\lceil \frac{\eta_N(p)}{{{R}_{t}}}
\right\rceil , \text{ if $p$ is an intersection point,}}
_{max\left(0,\left\lceil \frac{\eta_N(p)-{{R}_{s}}}{{{R}_{t}}} \right\rceil\right) ,
\text{ otherwise}.} \right.
\end{equation}
Note that
when $p$ is some reference point ${p}_{t_z}$ in $X$ and $\eta_N(p)$ $\leq$ ${{R}_{s}}$,
$\phi_N (p)$ is set to $0$. This is because $t_z$ can be directly covered by some potential points in $N$.

In the WMCBA, a greedy algorithm is applied to iteratively select
the set with the maximum weight of the uncovered elements for the RMWTCSCLMS
problem. In the MWTCSCLMS problem, because the number of mobile
sensors is limited, the mobile sensors have to be efficiently
utilized, and therefore, the idea of the STBA is to iteratively
select a set of targets $P_p$ for some $p \in X$ that has the maximum
weight of the uncovered targets and requires the minimum number of additional
potential points for constructing a network connected with the data
sink to cover the targets in $P_p$. Therefore, when $N$ is given, a new metric for each $p \in X$, denoted by $\rho_{N}
(p)$, is defined in Eq. \ref{profit}:
\begin{equation}\label{profit}
\rho_N (p)=\left\{^{\infty , \text{ if $\phi_N (p)$ $=$ $0$,}}
_{\frac{\Omega_N(P_p)}{\phi_N (p)},
\text{ otherwise},} \right.
\end{equation}
where $\Omega_N(P_p)$ denotes the total weight of the targets in $P_p$
that are not covered by the potential points in $N$.

In the STBA, the idea is to iteratively select a $p$ with higher
$\rho_N(p)$ and generate potential points to construct a connected
network and cover $P_p$. When the number of potential points is
higher such that the next $p$ is hard to select due to the
limited mobile sensors, a node-weighted Steiner tree algorithm is
applied to try to re-generate and minimize the number of potential
points. When all targets are covered or no more potential points can
be reduced, the potential points are determined. Then, similar to
the WMCBA, the Hungarian method is applied for assigning mobile
sensors to the potential points to minimize the total movement
distance of the mobile sensors. The STBA is described in detail in
Algorithm \ref{alg:STBA}. In Algorithm \ref{alg:STBA}, the $X$ is
constructed in Lines $1$-$3$. In addition, $L$ and $Y$ are
initialized to be $\emptyset$, where $L$ is used to store the generated
potential points and $Y$ is used to store the selected $p \in X$. In
the inner while loop, the point $p$, which is in $X$ such that at least one of the targets in $P_p$ is
not yet covered by the points in $Y$, with higher $\rho_{N} (p)$ is
iteratively selected, where $N$ is the union of $\{sink\}$ and $L$.
If two or more $p$ have the same $\rho_{N} (p)$, the $p$
with the lowest $\eta_N(p)$ is selected. Once a $p$ is selected, the
potential points can be therefore generated by $p$ and added into
$L$ to form a connected network. Therefore, $N$ can become a
bigger connected network when more potential points are generated.
When the number of potential points is higher such that no more $p$
can be selected, it breaks the inner while loop and calls the function
ReGeneratePotentialPoints. The function ReGeneratePotentialPoints is
to re-generate potential points to cover all targets in $P_{p_y}$
for each $p_y \in Y$ such that the number of the required potential
points is minimized, which is discussed later. The outer while loop
iteratively executes the inner while loop until all targets are
covered or no more potential points can be reduced. Finally, when
$L$ is determined, the deployment orders that assign mobile sensors to
potential points can be generated with the Hungarian method. The cost matrix
used for the Hungarian method is shown as follows:
\[{{[\psi(s_i,p_j)]}_{n\times n}}=\left(\begin{matrix}
   {\psi(s_1,p_1)} & \ldots  & {\psi(s_1,p_\ell)} & 0 & \cdots  & 0 \\
   \vdots  & \ddots  & \vdots & \vdots  & \ddots  & \vdots \\
   {\psi(s_n,p_1)} & \cdots  & {\psi(s_n,p_\ell)} & 0 & \cdots  & 0 \\
\end{matrix} \right),\]
where $n$ $=$ $|S|$; $\ell$ $=$ $|L|$; $\psi(s_i,p_j)$ $=$ $\zeta (s_i,p_j)$ if $p_j$ is in $L$; and
$\psi(s_i,p_j)$ $=$ $0$ for other cases.

Take Fig. \ref{Fig:exam_reference_points}, for example. When the
reference points in $X$ are obtained, the STBA then
iteratively selects a reference point $p$ with higher $\rho_N(p)$,
and generate potential points to construct a connected network and
cover $P_p$. Assume that $p^2_{t_1,t_2}$ has a higher $\rho_N$ value
than the other reference points and is selected. Because $L$ $=$
$\emptyset$ and only $sink$ is in $N$, the $sink$ in $N$ has the
shortest distance to $p^2_{t_1,t_2}$. Then, the potential points are
generated on the straight line from $sink$ to $p^2_{t_1,t_2}$ every
distance $R_t$, not including the $sink$, until $p^2_{t_1,t_2}$ is
reached or all targets in $P_{p^2_{t_1,t_2}}$ are covered. Clearly,
potential points $p_1$, $p_2$, and $p_3$ are generated accordingly.
In addition, $p^2_{t_1,t_2}$ is inserted into $Y$. After some
iterations, assume that $Y$ $=$ $\{p_{t_3}$, $p_{t_5}$, $p_{t_8}$,
$p^2_{t_1,t_2}$, $p^1_{t_6,t_7}\}$ and $14$ potential points are
generated as shown in Fig. \ref{Fig:exam_reference_points}. Because
the number of potential points is higher such that no more reference
points can be selected to cover the last target $t_4$, the function
ReGeneratePotentialPoints is called to re-generate potential points
to cover all targets in $P_{p_y}$ for each $p_y \in Y$ and to
minimize the required potential points. Assume that the potential
points re-generated by the function ReGeneratePotentialPoints are
shown in Fig. \ref{Fig:exam_regen_potential_points}. It is clear
that only $12$ potential points are required at this time. Then,
the reference point $p_{t_4}$ can be selected, and two potential
points can be generated to cover the targets in $P_{p_{t_4}}$. Then,
the deployment orders can be generated with the Hungarian method, as
shown in Fig. \ref{Fig:exam_MWSN}.

\begin{figure}
\center \subfigure{\includegraphics[width=8.5cm]{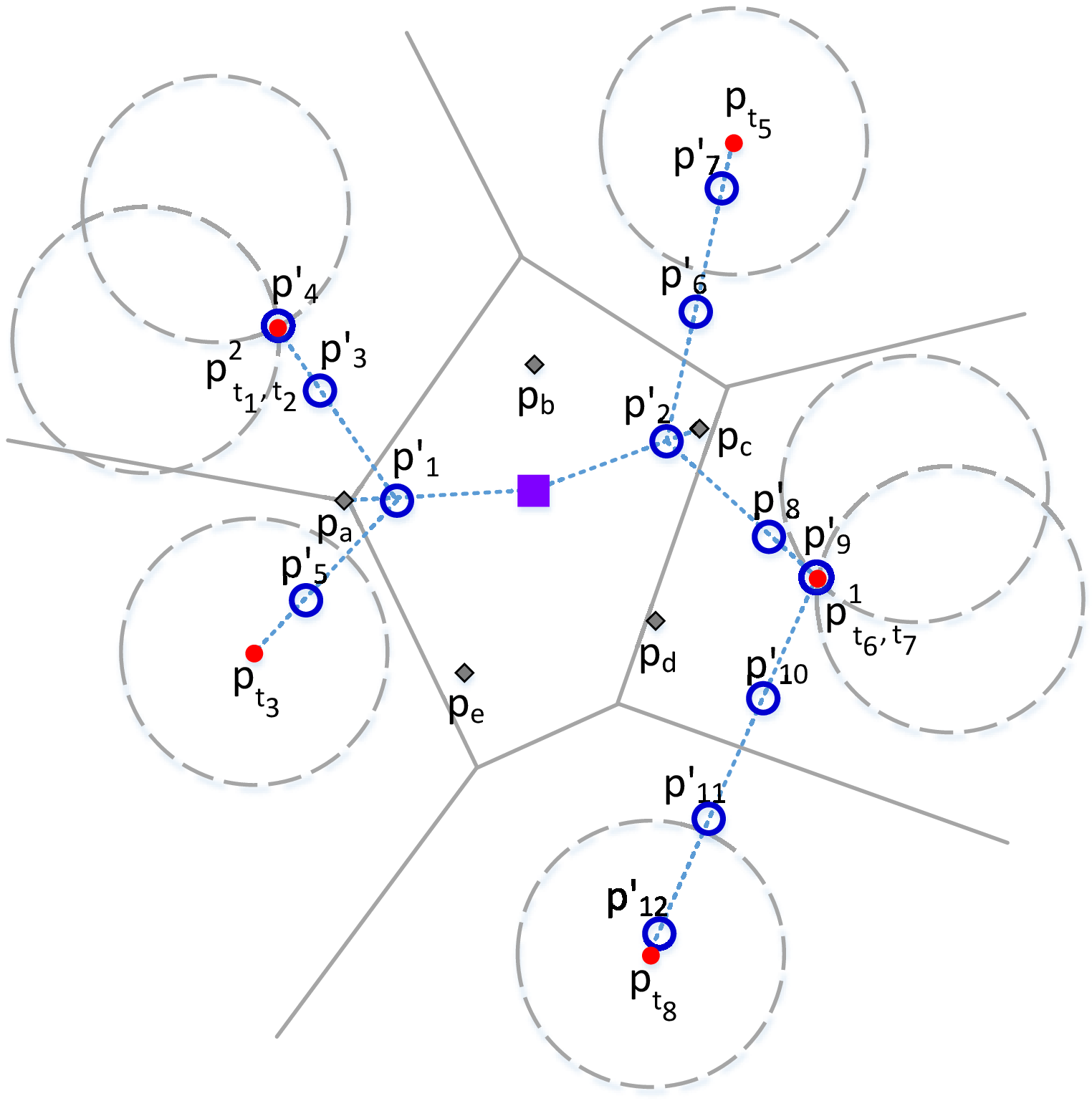}}\\
\subfigure{\includegraphics[width=6cm]{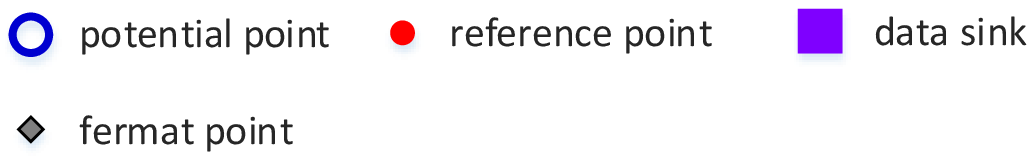}}
\caption{Example of the potential points re-generated by the
function ReGeneratePotentialPoints.}
\label{Fig:exam_regen_potential_points}
\end{figure}

\begin{algorithm} \caption{Steiner-Tree-Based Algorithm ($T$, $S$, $sink$)} \label{alg:STBA}
%\algsetup{linenosize=\small, linenodelimiter=.}
\begin{algorithmic}[1]

\State Let $X_1$ be the set of location points ${p}_{t_i}$ for all ${{t}_{i}}$
$\in$ $T$; and $P_{{p}_{t_i}}$ is set to $\{{t}_{i}\}$ for each ${p}_{t_i} \in X_1$

\State Let $X_2$ be the set of intersection points ${p}^1_{t_i,t_j}$ and
${p}^2_{t_i,t_j}$ for any $t_i,t_j$ $\in$ $T$; and $P_{p^1_{t_i,t_j}}$ and $P_{p^2_{t_i,t_j}}$ for any ${p}^1_{t_i,t_j}$, ${p}^2_{t_i,t_j}$ $\in$ $X_2$
are set to the sets of targets within
circles with radii ${{R}_{s}}$ centered at ${p}^1_{t_i,t_j}$ and ${p}^2_{t_i,t_j}$, respectively

\State $X$ $\gets$ $X_1$ $\cup$ $X_2$

\State $L$ $\gets$ $\emptyset$; $Y$ $\gets$ $\emptyset$

\While {there exists one target $t \in T$ not within any circles centered at $p \in L$ with radii $R_s$}

    \While{there exists one target $t \in T$ not within any circles centered at $p \in L$ with radii $R_s$}

        \State $N$ $\gets$ $\{sink\}$ $\cup$ $L$

        \State Let $Q$ be the set of points $p$ for all $p$ $\in$ $X$ with $P_p$ $-$ $\bigcup_{p_y \in Y}{P_{p_y}}$ $\neq$ $\emptyset$ and $\phi_N ({p})$ $\le$ $|S|$ $-$ $|L|$

        \If {$|Q| > 0$}

            \State Select a point $p_i$ from $Q$ such that $\rho_N (p_i)$ $>$ $\rho_N (p_j)$ or ($\rho_N (p_i)$ $=$ $\rho_N (p_j)$ and $\eta_N(p_i)$ $\le$ $\eta_N(p_j)$) for all ${p_j}$ $\in$ $Q$ $-$ $\{p_i\}$

            \State Let ${p_k}$ be the point in $N$ that has the shortest distance to ${p_i}$

            \State Generate potential points on the straight line from ${p_k}$ to ${p_i}$ every distance $R_t$, not including ${p_k}$, until ${p_i}$ is reached or all targets in $P_{p_i}$ are covered, and add the points to $L$

            \State $Y$ $\gets$ $Y$ $\cup$ $\{{p}_{i}\}$

        \Else

            \State break

        \EndIf

    \EndWhile

    \State $L'$ $\gets$ $ReGeneratePotentialPoints(Y, sink)$

    \If {$\left| L' \right|$ $<$ $\left| L \right|$}

        \State $L$ $\gets$ $L'$

    \Else

        \State break

    \EndIf

\EndWhile

\State Let $D$ be a set of deployment orders generated by the
Hungarian method \cite{NAV:NAV3800020109} with input $L$ and $S$

\State \Return $D$

\end{algorithmic}
\end{algorithm}

While given $Y$, which is used to store the selected reference
points, the goal of the function ReGeneratePotentialPoints is
designed to re-generate potential points to cover all targets in
$P_{p_y}$ for each $p_y \in Y$ and form a network connected with the
data sink such that the number of required potential points is
minimized. For this purpose, our idea is to find a tree in the plane
to connect each $p_y \in Y$ and the data sink $sink$, in which each
tree node is either a $p_y \in Y$ or another point, called the
intermediate point hereafter, in the plane, such that the total
length of the edges in the tree is minimized. To find more suitable
intermediate points in the plane, a set of points $F$ is
constructed by finding the Fermat points \cite{6023087} for $Y$ and
$sink$ in the plane. Then, we transfer to the Node-Weighted Steiner
Tree (NWST) problem to find a tree $\Upsilon$ that spans the data
sink, each $p_y \in Y$, and some points in $F$, such that the total
length of the edges in the tree is minimized. When the tree $\Upsilon$
is obtained, the potential points are re-generated by the tree
structure.

While given three vertices (or points) of a triangle $\Delta$, the
Fermat point $p$ is a point in the plane such that the total
distance from each of the three vertices to $p$ is the minimum. The
Fermat point can be obtained with the following rule. If the triangle
$\Delta$ has an angle not less than $120$ degrees, the Fermat point
is located at the obtuse angled vertex of the $\Delta$. Otherwise, we
can construct an equilateral triangle on each of any two sides of
the $\Delta$. Then, the Fermat point is located at the point
intersected by the two lines that are drawn from each new vertex to
the opposite vertex of the $\Delta$. To find usable intermediate
points in the plane, a set of the Fermat points $F$ is constructed
for the $Y$ and the $sink$ as follows. While given $Y$ and $sink$, a
Voronoi diagram for the points in $Y$ $\cup$ $\{sink\}$ is first
constructed. The Voronoi diagram for $Y$ $\cup$ $\{sink\}$ is the
polygonal partition of the plane. In addition, each polygon $Z(p)$
is associated with a point in $Y$ $\cup$ $\{sink\}$ such that all
points in $Z(p)$ are closer to $p$ than other points in $Y$ $\cup$
$\{sink\}$. Two points $p_i$, $p_j$ $\in$ $Y$ $\cup$ $\{sink\}$ are
said to be neighbors in the Voronoi diagram if $Z(p_i)$ and $Z(p_j)$
share a common boundary in the Voronoi diagram. By the generated
Voronoi diagram, $F$ can be constructed by finding the Fermat points
for any three points $p_i$, $p_j$, $p_z$ $\in$ $Y$ $\cup$ $\{sink\}$
that are neighbors to each other. Take Fig.
\ref{Fig:exam_regen_potential_points}, for example. Assume that $Y$
$=$ $\{p_{t_3}$, $p_{t_5}$, $p_{t_8}$, $p^2_{t_1,t_2}$,
$p^1_{t_6,t_7}\}$. By $Y$ and the $sink$, the Voronoi diagram can be
constructed as shown in Fig. \ref{Fig:exam_regen_potential_points}.
In addition, because the $sink$, $p_{t_3}$, and $p^2_{t_1,t_2}$ are
neighbors to each other, the corresponding Fermat point can then be
constructed as the point $p_a$ in Fig.
\ref{Fig:exam_regen_potential_points}. In Fig.
\ref{Fig:exam_regen_potential_points}, five Fermat points, including
$p_a$, $p_b$, $p_c$, $p_d$, and $p_e$, are generated.

Here, we show how to transfer to the NWST problem to find a tree
$\Upsilon$, which spans the $sink$, each $p_y \in Y$, and some
points in $F$, such that the total length of the edges in the tree is
minimized. In the NWST problem, when given an undirected weighted
graph $G(V_G,E_G,\kappa)$ and a set of terminal nodes $TS$, the
problem is to find a tree $\Upsilon(V_\Upsilon,E_\Upsilon)$ in $G$
with $TS \subseteq V_\Upsilon$ and $V_\Upsilon \subseteq V_G$ such
that the total weight of the edges and nodes in $\Upsilon$ is minimized,
where $V_G$ (or $V_\Upsilon$) is a set of nodes, $E_G$ (or
$E_\Upsilon$) is a set of edges connecting two nodes in $V_G$ (or
$V_\Upsilon$), and $\kappa (v)$ (or $\kappa ((u,v))$) denotes the
weight of node $v \in V_G$ (or edge $(u,v)$ $\in$ $E_G$). Let $TS$
$=$ $Y$ $\cup$ $\{sink\}$ and $V$ $=$ $TS$ $\cup$ $F$. Also let
$G(V_G,E_G,\kappa)$ be a weighted complete graph generated by $V$,
where $V_G = V$, $E$ is the set of $(p_i, p_j)$ for any $p_i, p_j$
in $V$, $\kappa(p_i, p_j)$ is the distance between $p_i$ and $p_j$
for any $p_i, p_j$ in $V$, and $\kappa(p_i)$ $=$ $0$ for any $p_i$
in $V$. It is clear that when $sink$, $Y$, and $F$ are given, the
problem is to find a node-weighted Steiner tree with the minimum total
weight. Therefore, when $G$ and $TS$ are generated by $sink$, $Y$,
and $F$, a method, called the modified Klein and Ravi algorithm,
extended from the algorithm \cite{Klein1995104} proposed by Klein
and Ravi used for the NWST problem, is proposed to find a tree
$\Upsilon(V_\Upsilon,E_\Upsilon)$, which can span the $sink$, each
$p_y \in Y$, and some points in $F$, such that the total length of
the edges in $\Upsilon$ is minimized. The details of the Klein and
Ravi algorithm are described as follows. In the Klein and Ravi
algorithm, initially, each terminal node in $TS$ is in a tree by
itself. Then, the trees are iteratively selected and merged into a
bigger tree until only one tree is left. Let $\Gamma$ be the set of
all trees, and let $\xi (v,\Upsilon_i)$ denote the minimum sum of
the weights of the nodes and edges in the path from $v$ to the tree
$\Upsilon_i$, excluding its endpoints. The quotient cost of a node
$v$ is defined in Eq. \ref{eq:quotient_cost}:
\begin{equation}\label{eq:quotient_cost}
\min_{\Gamma' \subseteq \Gamma,|\Gamma'| \geq
2}{\frac{\kappa(v)+\sum_{\Upsilon_i \in \Gamma'}{\xi
(v,\Upsilon_i)}}{|\Gamma'|}}.
\end{equation}
In each iteration of mergence, the node with the minimum quotient cost
is first selected. Then, the corresponding paths and trees
selected in evaluating the quotient cost are merged into one tree.
Take a weighted complete graph with four nodes $p_1$, $p_2$, $p_3$,
and $p_4$, for example, where $p_1$, $p_2$, and $p_3$ are terminal
nodes, $\kappa(p_1,p_2)$ $=$ $\kappa(p_1,p_3)$ $=$ $\kappa(p_2,p_3)$
$=$ $10$, $\kappa(p_1,p_4)$ $=$ $\kappa(p_2,p_4)$ $=$
$\kappa(p_3,p_4)$ $=$ $\frac{10}{\sqrt{3}}$, and $\kappa(p_1)$ $=$
$\kappa(p_2)$ $=$ $\kappa(p_3)$ $=$ $\kappa(p_4)$ $=$ $0$.
Initially, each of $p_1$, $p_2$, and $p_3$ is in a tree by itself.
Let $\Upsilon_1$, $\Upsilon_2$, and $\Upsilon_3$ be the trees
that include only $p_1$, $p_2$, and $p_3$, respectively. Clearly, the
quotient cost of $p_1$, $p_2$, or $p_3$ is $\frac{10}{2}$ $=$ $5$,
and the quotient cost of $p_4$ is $\frac{\frac{10}{\sqrt{3}} \times
3}{3}$ $=$ $\frac{10}{\sqrt{3}}$ $>$ $5$. Therefore, $\Upsilon_1$
and $\Upsilon_2$ will be merged into a bigger tree $\Upsilon_{1,2}$
by inserting an edge $(p_1,p_2)$. Finally, $\Upsilon_3$ and
$\Upsilon_{1,2}$ will be merged into a final tree by inserting an
edge $(p_1,p_3)$. It is clear that the optimum solution for this
case is a tree with nodes $p_1$, $p_2$, $p_3$, $p_4$ and edges
$(p_1,p_4)$, $(p_2,p_4)$, $(p_3,p_4)$. To achieve this, the modified
Klein and Ravi algorithm is therefore proposed here with a
modification of the definition of the quotient cost. In the modified
Klein and Ravi algorithm, the quotient cost of a node $v$ is
defined in Eq. \ref{eq:modified_quotient_cost}:
\begin{equation}\label{eq:modified_quotient_cost}
\min_{\Gamma' \subseteq \Gamma,|\Gamma'| \geq
2}{\frac{\kappa(v)+\sum_{\Upsilon_i \in \Gamma'}{\xi
(v,\Upsilon_i)}}{|\Gamma'| - 1}}.
\end{equation}
In the modified Klein and Ravi algorithm, the quotient cost of
$p_1$, $p_2$, or $p_3$ is $\frac{20}{2}$ $=$ $10$, and the quotient
cost of $p_4$ is $\frac{\frac{10}{\sqrt{3}} \times 3}{2}$ $=$
$5\sqrt{3}$ $<$ $10$. Therefore, $\Upsilon_1$, $\Upsilon_2$, and
$\Upsilon_3$ will be merged into a tree with inserting edges
$(p_1,p_4)$, $(p_2, p_4)$, and $(p_3, p_4)$.

Take Fig. \ref{Fig:exam_regen_potential_points} as another example.
When the data sink $sink$, the reference points $p_{t_3}$,
$p_{t_5}$, $p_{t_8}$, $p^2_{t_1,t_2}$, $p^1_{t_6,t_7}$, and the
Fermat points $p_a$, $p_b$, $p_c$, $p_d$, $p_e$ are given, the
weighted complete graph can be constructed accordingly as shown in
Fig. \ref{Fig:exam_steiner_tree}. In Fig.
\ref{Fig:exam_steiner_tree}, the $sink$ and the five reference
points are terminal nodes. By the modified Klein and Ravi
algorithm, a tree that spans the $sink$ and the five reference
points is constructed, where the edges in the tree include
$(sink,p_a)$, $(sink,p_c)$, $(p_a,p_{t_3})$, $(p_a,p^2_{t_1,t_2})$,
$(p_c,p_{t_5})$, $(p_c,p^1_{t_6,t_7})$, and $(p^1_{t_6,t_7},
p_{t_8})$.

\begin{figure}
\center \subfigure{\includegraphics[width=8cm]{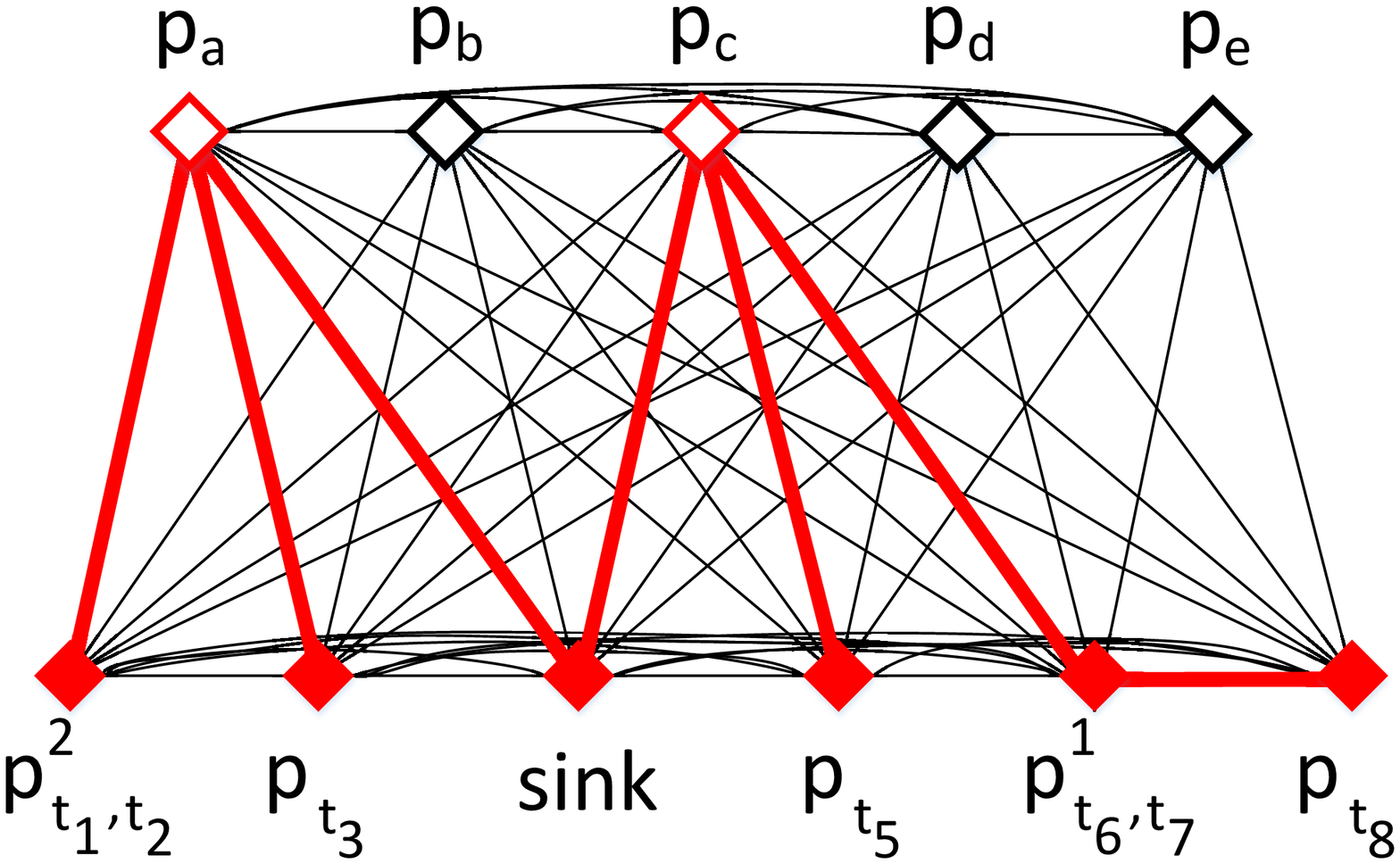}}\\
\subfigure{\includegraphics[width=6cm]{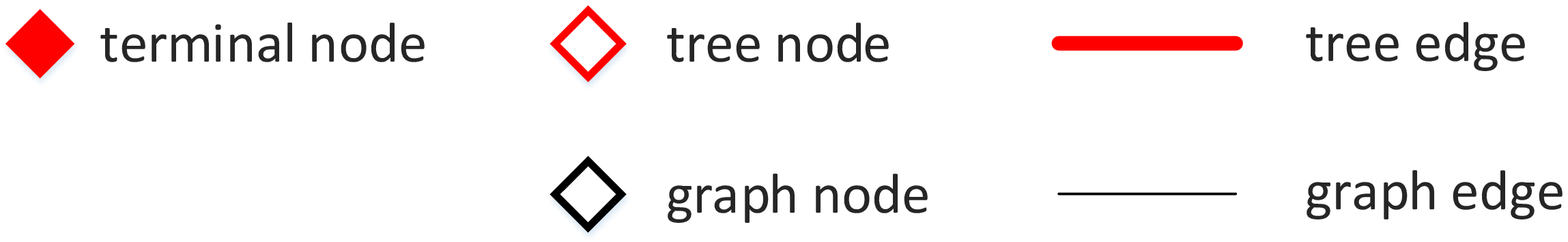}}
\caption{Example of constructing a tree to span all terminal nodes
with the modified Klein and Ravi algorithm, where the terminal nodes
include $sink$, $p_{t_3}$, $p_{t_5}$, $p_{t_8}$, $p^2_{t_1,t_2}$,
and $p^1_{t_6,t_7}$.} \label{Fig:exam_steiner_tree}
\end{figure}

When the tree $\Upsilon(V_\Upsilon,E_\Upsilon)$ is obtained, our
idea is to deploy potential points along the paths from the data
sink, through tree edges, to cover all targets in $P_{p_y}$ for each
$p_y \in Y$. Let $p.\gamma$ be a point for any tree node $p \in
V_\Upsilon$ that can represent $p$ to connect to the potential
points in the other tree edges. Initially, $p.\gamma$ is set to $sink$
if $p$ $=$ $sink$; otherwise, $p.\gamma$ is initialized to $null$.
For any tree edge $(p_i,p_j) \in E_\Upsilon$ with $p_i.\gamma \neq
null$, we deploy potential points on the straight line from
$p_i.\gamma$ to $p_j$ every distance $R_t$, not including
$p_i.\gamma$. If $p_j$ is not a Fermat point in $F$, the potential
points are generated until ${p_j}$ is reached or all targets in
$P_{p_j}$ are covered; otherwise, the potential points are generated
until ${p_j}$ is reached or the last generated potential point can
cover ${p_j}$. The ${p_j}.\gamma$ is then set to the last generated
potential point. In addition, the generated potential points are
recorded. The process is repeated until all tree edges are
referenced to generate potential points. The details can be seen in
the function ReGeneratePotentialPoints. Take Fig.
\ref{Fig:exam_steiner_tree}, for example. When the tree in Fig.
\ref{Fig:exam_steiner_tree} is obtained, edge $(sink,p_a)$ or edge
$(sink,p_b)$ is selected to generate potential points because
$sink.\gamma = sink$. Assume that edge $(sink,p_a)$ is selected
first. Because $p_a$ is a Fermat point, the potential points are
generated on the straight line from $sink$ to $p_a$ every distance
$R_t$, not including $sink$, until $p_a$ is reached or the last
generated point can cover $p_a$. As shown in Fig.
\ref{Fig:exam_regen_potential_points}, it is clear that only
potential point $p'_1$ is generated because $p'_1$ can cover $p_a$.
In addition, $p_a.\gamma$ is set to $p'_1$. Assume that edge
$(p_a,p^2_{t_1,t_2})$ is selected later. Because $p^2_{t_1,t_2}$ is
not a Fermat point, the potential points are generated on the
straight line from $p_a.\gamma$ ($=$ $p'_1$) to $p^2_{t_1,t_2}$
every distance $R_t$, not including $p'_1$, until $p^2_{t_1,t_2}$ is
reached or all targets in $P_{p^2_{t_1,t_2}}$ are covered. As shown
in Fig. \ref{Fig:exam_regen_potential_points}, clearly, potential
points $p'_3$ and $p'_4$ are generated, and $p'_4$ is the last
generated potential point for edge $(p_a,p^2_{t_1,t_2})$ because
$p'_4$ is located at $p^2_{t_1,t_2}$. Using the same process, $12$
potential points can be generated as in Fig.
\ref{Fig:exam_regen_potential_points}.

\begin{algorithm} \label{alg:STBA_R}
%\algsetup{linenosize=\small, linenodelimiter=.}
\begin{algorithmic}[1]

\Function {ReGeneratePotentialPoints} {$Y$, $sink$}

\State $F$ $\gets$ $\emptyset$

\State Generate a Voronoi diagram for $Y$ $\cup$ $\{sink\}$

\For {any three points $p_i$, $p_j$,
$p_k$ $\in$ $Y$ $\cup$ $\{sink\}$ that are neighbors to each other in
the generated Voronoi diagram}

    \State Let $p$ be the node located at the Fermat point generated by $p_i$, $p_j$, and $p_k$

    \State $F$ $\gets$ $F$ $\cup$ $\{p\}$

\EndFor

\State $TS$ $\gets$ $Y$ $\cup$ $\{sink\}$; $V$ $\gets$ $TS$ $\cup$
$F$

\State Construct a weighted complete graph $G(V_G,E_G,\kappa)$ by
$V$

\State Construct a tree $\Upsilon(V_\Upsilon,E_\Upsilon)$ by the
modified Klein and Ravi algorithm with input $G$ and $TS$

\For {each $p \in V_\Upsilon$}

\State $p.\gamma$ $\gets$ $sink$ if $p$ $=$ $sink$; otherwise,
$p.\gamma$ $\gets$ $null$

\EndFor

\State $L'$ $\gets$ $\emptyset$

\While {$E_\Upsilon$ $\neq$ $\emptyset$}

    \For {each $(p_i, p_j)$ $\in$ $E_\Upsilon$ with $p_i.\gamma$ $\neq$ $null$}

        \If {${p_j}$ is not a Fermat point}

            \State Generate potential points on the straight line from ${p_i.\gamma}$ to ${p_j}$ every distance $R_t$, not including ${p_i.\gamma}$, until ${p_j}$ is reached or all targets in $P_{p_j}$ are covered; and then set ${p_j}.\gamma$ to the last generated point and add all points to $L'$

        \Else

            \State Generate potential points on the straight line from ${p_i.\gamma}$ to ${p_j}$ every distance $R_t$, not including ${p_i.\gamma}$, until ${p_j}$ is reached or the last generated point can cover ${p_j}$; and then set ${p_j}.\gamma$ to the last generated point and add all points to $L'$

        \EndIf

        \State $E_\Upsilon$ $\gets$ $E_\Upsilon$ $-$ $\{(p_i, p_j)\}$

    \EndFor

\EndWhile

\State \Return $L'$

\EndFunction

\end{algorithmic}
\end{algorithm}

The time complexity of the STBA is provided in Theorem
\ref{thm:STBA_time_complexity}.

\begin{thm}\label{thm:STBA_time_complexity}
The time complexity of the STBA is bounded in $O(m^5 + n^3)$, where
$m$ is the number of targets and $n$ is the number of mobile
sensors.
\end{thm}

\begin{proof}
In the function ReGeneratePotentialPoints, when $Y$ and $sink$ are
given, it requires at most $O((|Y|+1)^3)$ $=$ $O(|Y|^3)$ time to
find the Fermat points because any three neighboring points in $Y
\cup \{sink\}$ in the Voronoi diagram have to be checked. By
\cite{Lee1980}, we have that
at most $2 \times (|Y|+1)$ combinations of three neighboring points
in $Y \cup \{sink\}$ in the Voronoi diagram, and thus, the generated
weighted complete graph has at most $3 \times (|Y|+1)$ nodes. In the
modified Klein and Ravi algorithm, each node in the weighted
complete graph has to compute its distances to all trees in each
iteration \cite{Klein1995104}, and thus, each iteration
requires at most $O((3 \times (|Y|+1))^3)$ time. Because at least
two trees are merged into one tree in each iteration, at most $|Y|$
iterations are required, and thus, the modified Klein and Ravi
algorithm requires at most $O(|Y| \times (3 \times (|Y|+1))^3)$ $=$
$O(|Y|^4)$ time. Because at most $2 \times {\left(\begin{matrix} 3\times(|Y|+1) \\
2 \\ \end{matrix} \right)}$ edges in $E_\Upsilon$,
it requires at most $O(|Y|^2)$ time to generate potential points.
Therefore, the function ReGeneratePotentialPoints requires at most
$O(|Y|^3)$ $+$ $O(|Y|^4)$ $+$ $O(|Y|^2)$ $=$ $O(|Y|^4)$ time.

In the STBA, because there are $m$ targets in $T$, it
requires at
most $O(m^2)$ time to generate $2 \times {\left(\begin{matrix} m \\
2 \\ \end{matrix} \right)}$ $=$ $m^2-m$ intersection points, and
thus, it requires $O(m \times m^2)$ $=$ $O(m^3)$ time to compute
$P_{p^1_{t_i,t_j}}$ and $P_{p^2_{t_i,t_j}}$ for any
${p}^1_{t_i,t_j}$, ${p}^2_{t_i,t_j}$ $\in$ $X_2$. In addition, it is
clear that it requires $O(m)$ time to compute $X_1$. Because $X$ $=$
$X_1$ $\cup$ $X_2$, it requires $O(m^3)$ $+$ $O(m)$ $=$ $O(m^3)$
time to compute $X$. In each iteration of the outer while loop, at
least one target will be covered in the inner while loop, or the
function ReGeneratePotentialPoints is called once. Because at most
$(m^2-m)$ $+$ $m$ $=$ $m^2$ elements exist in $X$, at most $m^2$
elements in $X$ have to be checked in the inner while loop. The
inner while loop requires at most $O(m^2 \times m^2)$ $=$ $O(m^4)$
time because each element $p$ in $X$ requires at most $O(m^2)$ time
to verify that $P_p$ $-$ $\bigcup_{p_y \in Y}{P_{p_y}}$ $\neq$
$\emptyset$ and $\phi_N ({p})$ $\le$ $|S|$ $-$ $|L|$. In addition,
the function ReGeneratePotentialPoints requires at most $O(m^4)$
time because at most $m$ reference points are included in $Y$ to
cover targets. Therefore, each iteration of the outer while loop
requires at most $O(m^4)$ $+$ $O(m^4)$ $=$ $O(m^4)$ time. Because at
least one target will be covered in each iteration of the outer
while loop, except for the final iteration, at most $m+1$ iterations
are required in the outer while loop. Therefore, the outer while
loop requires at most $O((m+1) \times m^4)$ $=$ $O(m^5)$ time.
Because the $n \times n$ cost matrix is required for the Hungarian
method, the Hungarian method requires at most $O(n^3)$ time to
compute deployment orders \cite{NAV:NAV3800020109}. Therefore, the STBA
requires at most $O(m^3)$ + $O(m^5)$ + $O(n^3)$ $=$ $(m^5 + n^3)$,
which completes the proof.
\end{proof}

\section{Performance Evaluation}\label{section:Simulation}

In this section, simulations were used to evaluate the performance
of the STBA. In the simulations, $10$-$400$ mobile sensors and
$10$-$50$ targets were randomly deployed in a $600 \times 600$
square area, where the sensing range $R_s$ and the transmission
range $R_t$ of the mobile sensors were set to $20$. In addition, the
data sink was deployed at the center of the sensing field. Moreover, the value of $\omega$ of each target was randomly selected
from the interval $[1,10]$. In the following simulation, the results
were obtained by averaging $100$ data.

To demonstrate the performance of the STBA, the heuristic algorithm,
called the target-based Voronoi greedy algorithm $+$ Euclidean
minimum spanning tree-Hungarian algorithm (TV-Greedy+ECST-H) was
compared. The TV-Greedy+ECST-H is used for the problem of scheduling
mobile sensors to cover all targets such that the total movement
distance of the mobile sensors is minimized. The TV-Greedy+ECST-H uses
targets' locations to divide the sensing field into Voronoi
partitions, which also divides mobile sensors into independent
groups. Each target is covered by the nearest sensor selected from
the target's group or the target's neighboring groups. Then,
a Euclidean minimum spanning tree is adopted to determine the
connected paths to the data sink such that mobile sensors can be
deployed on the paths. Because the TV-Greedy+ECST-H can be used only
for the MWSN with enough mobile sensors to cover all targets and
form a connected network, a heuristic, termed the greedy-based
algorithm (GBA), is thus proposed here for the MWTCSCLMS problem. In
the GBA, the idea is to iteratively select an adaptive target $t_i$
from $T$, deploy a potential point at the location of the $t_i$,
that is, $p_{t_i}$, and form a bigger network connected with the
data sink and the $p_{t_i}$, until there are not enough mobile sensors or
all targets are covered. Let $N_{GBA}$ be a set of the points
$p_{t_j}$ located at the selected targets $t_j$ and the data sink.
Also let $\eta_{N_{GBA}}(p)$ denote the minimum distance between
point $p$ and each point in $N_{GBA}$. When an adaptive target $t_i$
is selected to form a bigger network connected with $N_{GBA}$, the
GBA is to separate the straight line between ${{t}_{i}}$ and
${{t}_{min}}$ into $\left\lceil
\frac{\eta_{N_{GBA}}(p_{t_i})}{{{R}_{t}}} \right\rceil$ equal parts
by potential points, where $t_{min}$ denotes the target whose corresponding
location point in $N_{GBA}$ has the minimum distance to the
$p_{t_i}$. Here, let $\phi'_{N_{GBA}} (p)$ denote $\left\lceil
\frac{\eta_{N_{GBA}}(p)}{{{R}_{t}}} \right\rceil$. To find an
adaptive target, a new metric with a given $N_{GBA}$ for each target
$t_j \in T$ is therefore defined in Eq. \ref{eq:GBA:profit}:
\begin{equation}\label{eq:GBA:profit}
\rho'_{N_{GBA}} (t_j)= \frac{t_j.\omega}{\phi'_{N_{GBA}} (p_{t_j})}.
\end{equation}
In the GBA, $N_{GBA}$ is initialized to be $\{sink\}$. The selection
of an adaptive target is similar to selecting a $p$ with higher
$\rho_{N} (p)$ in the STBA. In each iteration, the target $t_i$ with
higher $\rho'_{N_{GBA}} (t_i)$ is selected, and the corresponding
potential points are generated to form a connected network with the
data sink and the $p_{t_i}$. If two or more targets $t$ have the
same $\rho'_{N_{GBA}} (t)$, the $t$ with lowest
$\eta_{N_{GBA}}(p_t)$ is selected. Then, the $p_{t_i}$ is
inserted into $N_{GBA}$. The process is iteratively executed until
the mobile sensors are not enough to select any target or all targets
are selected. When the potential points are determined, the cost
matrix is generated in the same way as that in the STBA, which is
used for the Hungarian method \cite{NAV:NAV3800020109} to generate
deployment orders.

To compare the STBA with the TV-Greedy+ECST-H, the WMCBA, and the
GBA, three MWSN scenarios were considered in the simulation. In the
first MWSN scenario, enough mobile sensors were provided such that
all targets can be fully covered and form a connected network, where
$200$-$400$ mobile sensors were randomly deployed in the sensing
field, the value of $\omega$ of each target was set to $1$, and the
TV-Greedy+ECST-H, the GBA, and the STBA could work here. In the second
MWSN scenario, the MWSN was the same as that in the RMWTCSCLMS
problem; that is, there may not be enough mobile sensors to cover all
targets, but the transmission range was large enough such that any
two mobile sensors (or any mobile sensor and the data sink) could
communicate with each other, where $10$-$30$ mobile sensors were randomly deployed in the sensing field, $R_t$
was set to $\infty$, and the WMCBA, the GBA, and the STBA could work
here. In the third scenario, the MWSN was the same as that in the
MWTCSCLMS problem, where $25$-$175$ mobile sensors were randomly
deployed in the sensing field, and the GBA and the STBA could work
here. In addition, the TV-Greedy+ECST-H, the WMCBA, the GBA, and the
STBA were compared in terms of the total number of mobile
sensors used, the total movement distance, and the total weight of the
covered targets. The first, second, and third MWSN scenarios
are discussed in Section \ref{section:sim:fullycovered}, Section
\ref{section:sim:RMWTCSCLMS}, and Section
\ref{section:sim:MWTCSCLMS}, respectively.

\subsection{Dense MWSNs}\label{section:sim:fullycovered}

\begin{figure}
\center
\subfigure[]{\includegraphics[width=4.25cm]{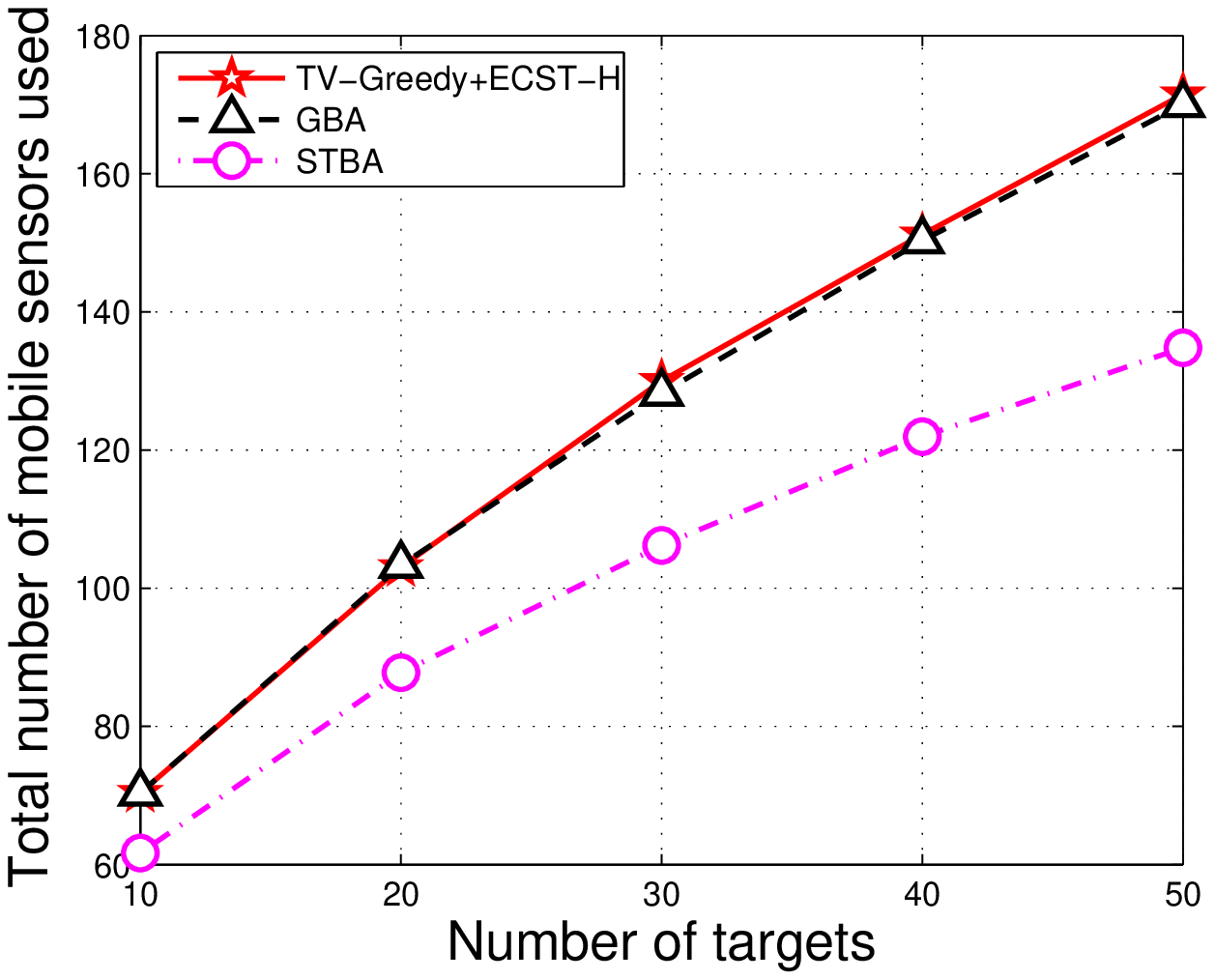}
\label{Fig:target_num_sensor}}
\subfigure[]{\includegraphics[width=4.25cm]{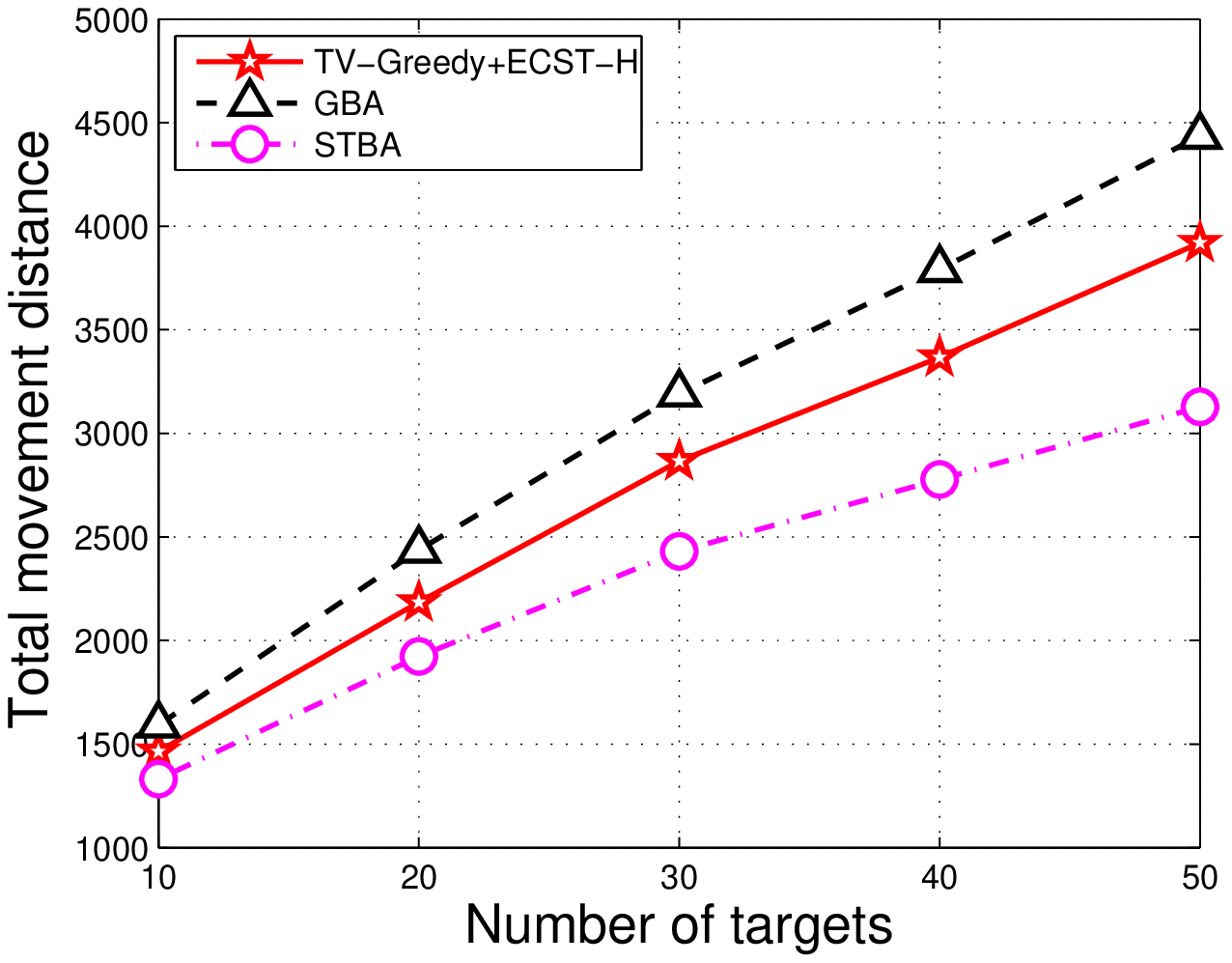}
\label{Fig:target_distance_full}}
\caption{The total number of mobile sensors used and the total movement distance required in MWSNs whose number of targets ranges from $10$ to $50$. The required total number of mobile sensors used and the total movement distance are shown in (a) and (b), respectively.}
\end{figure}

In dense MWSNs, unless otherwise stated, the number of targets
was set to $30$; and the number of mobile sensors was set to $300$. Fig. \ref{Fig:target_num_sensor} and Fig.
\ref{Fig:target_distance_full} show the comparisons of the total
number of mobile sensors used and the total movement distance,
respectively, in MWSNs when the number of targets ranges from $10$
to $50$. In Fig. \ref{Fig:target_num_sensor}, it is clear that the
higher the number of targets, the higher the total number of mobile
sensors used by the TV-Greedy+ECST-H, the GBA, and the STBA. This is
because more mobile sensors are required to cover targets and form a
connected network. Note that the STBA has a lower total number of
mobile sensors used than the TV-Greedy+ECST-H and the GBA. This is
because all possible sets of targets that can be covered by any
point in a sensing field are considered in the STBA such that
multiple targets have a high probability of being selected and covered by
only one mobile sensor to minimize the number of required mobile
sensors. In addition, the potential points can be re-generated by
the function ReGeneratePotentialPoints such that the network
connectivity is maintained and the number of required potential
points is reduced as much as possible. In Fig.
\ref{Fig:target_distance_full}, the higher the number of targets,
the longer the total movement distance required by the
TV-Greedy+ECST-H, the GBA, and the STBA. This is because more
targets are required to be covered by mobile sensors such that
more total movement distance is required for mobile sensors to cover
targets and form a connected network. In addition, the STBA has a
shorter total movement distance than the TV-Greedy+ECST-H and the
GBA. This is because fewer mobile sensors are required
to cover targets and form a connected network, as observed in Fig.
\ref{Fig:target_num_sensor}. Also note that the GBA has a longer
total movement distance than the TV-Greedy+ECST-H. This is because
the targets in the GBA are also potential points to which mobile
sensors are required to move, and therefore, more total movement
distance is required.

\begin{figure}
\center
\subfigure[]{\includegraphics[width=4.25cm]{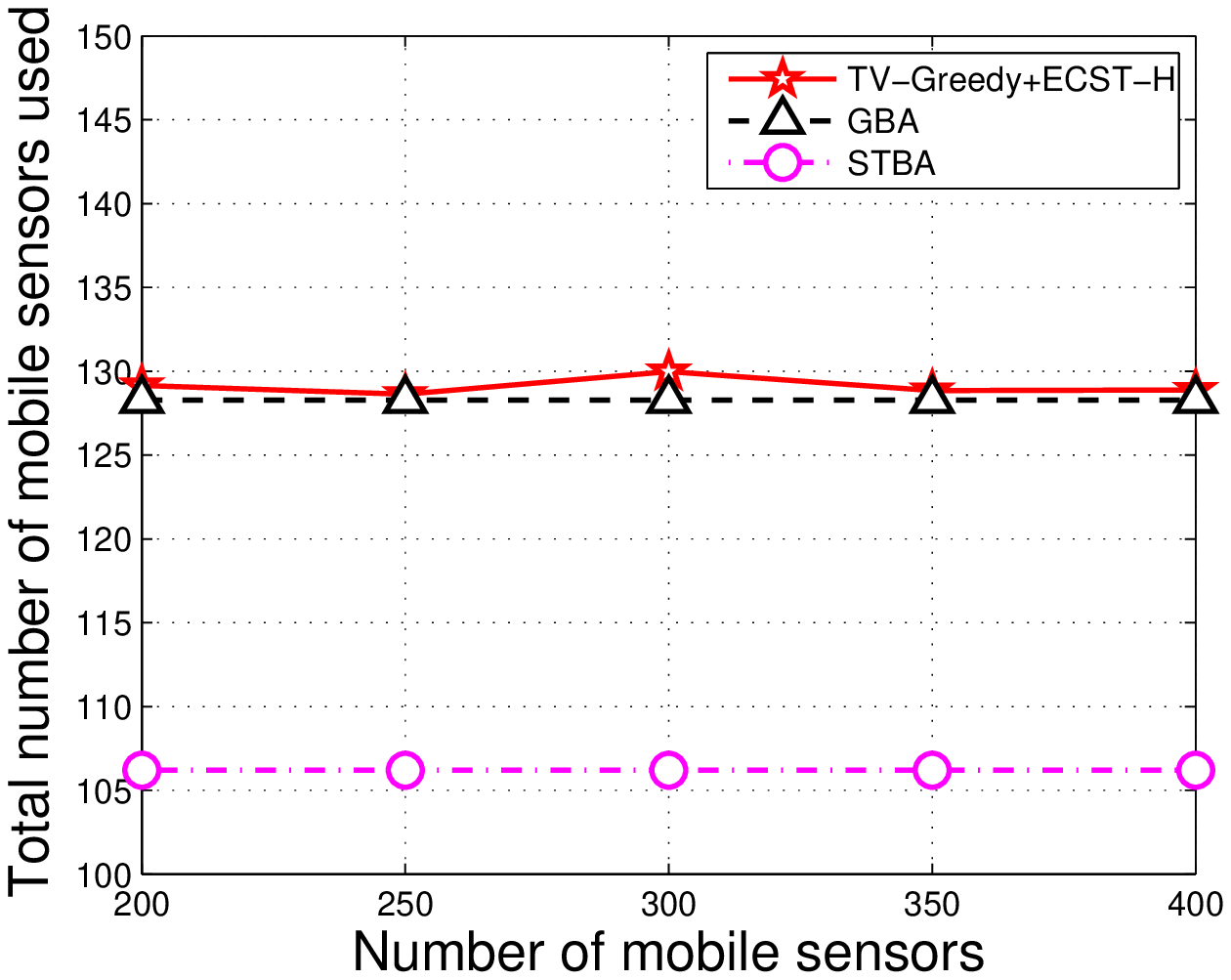}\label{Fig:sensor_sensor_full}}
\subfigure[]{\includegraphics[width=4.25cm]{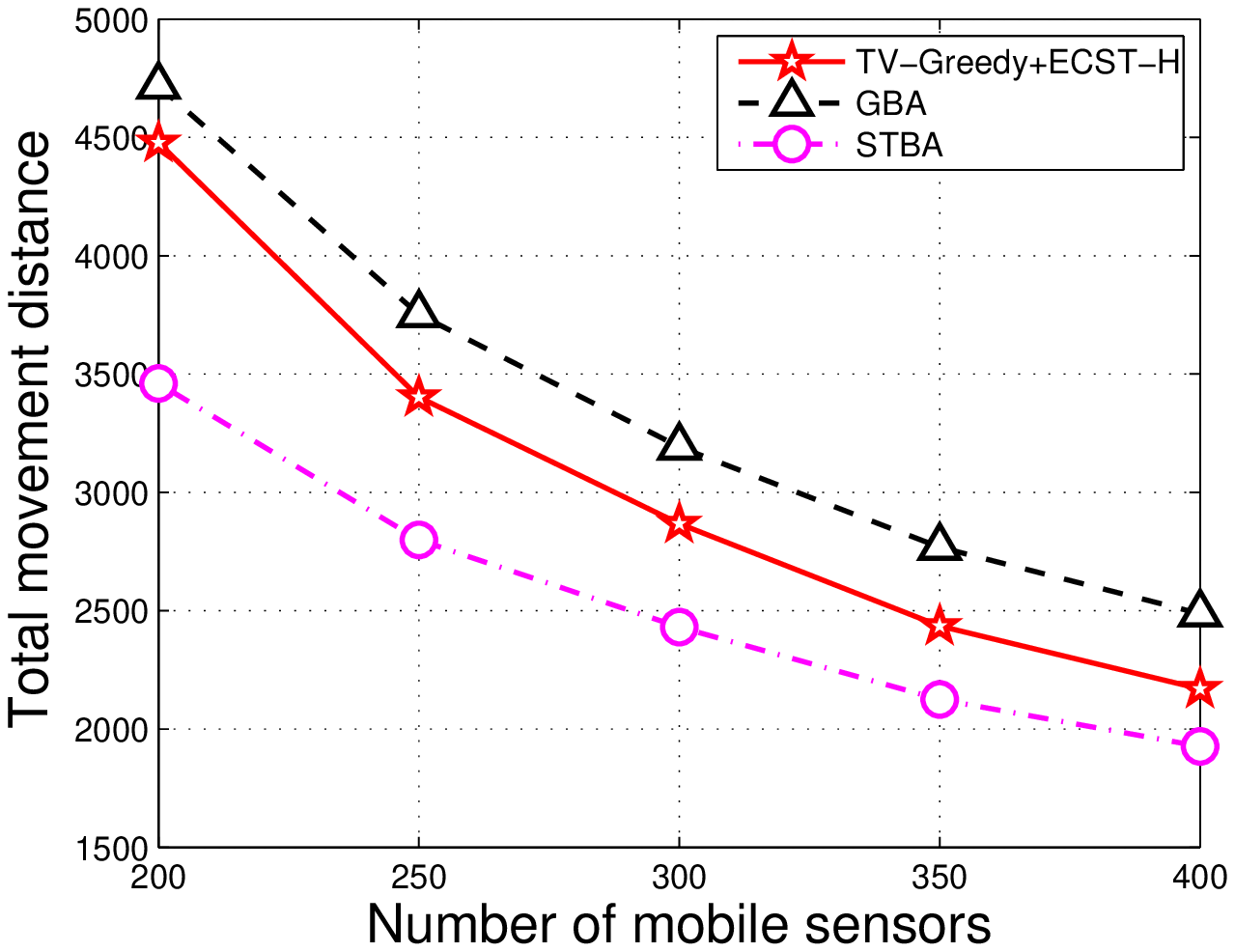}\label{Fig:sensor_distance_full}}
\caption{The total number of mobile sensors used and the total movement distance required in MWSNs whose number of mobile sensors ranges from $200$ to $400$. The required total number of mobile sensors used and the total movement distance are shown in (a) and (b), respectively.}
\end{figure}

Fig. \ref{Fig:sensor_sensor_full} and Fig. \ref{Fig:sensor_distance_full}
show the comparisons of the total number of mobile sensors used and the total movement distance, respectively, in MWSNs when
the number of mobile sensors ranges from $200$ to $400$.
In Fig. \ref{Fig:sensor_sensor_full}, it is clear that the TV-Greedy+ECST-H, the GBA, or the STBA has
similar results with the increasing number of mobile sensors.
This is because there are enough mobile sensors to cover $30$ targets and form a connected network.
In addition, the STBA requires the lowest number of mobile sensors used because the potential points generated by the STBA
are minimized to cover the targets and form a connected network, the same observation as in Fig. \ref{Fig:target_num_sensor}.
In Fig. \ref{Fig:sensor_distance_full}, the higher the number of mobile sensors, the lower the total movement distance required
by the TV-Greedy+ECST-H, the GBA, and the STBA. This stems from the fact that more nearby mobile sensors can be selected
to cover targets and form a connected network, and thus, the total movement distance of the mobile sensors is decreased.
It is clear that the STBA has a shorter total movement distance than the TV-Greedy+ECST-H and the GBA, as observed in Fig. \ref{Fig:target_distance_full}. This is because fewer mobile sensors are required to cover targets and form a connected network.

\begin{figure}
\center
\subfigure[]{\includegraphics[width=4.25cm]{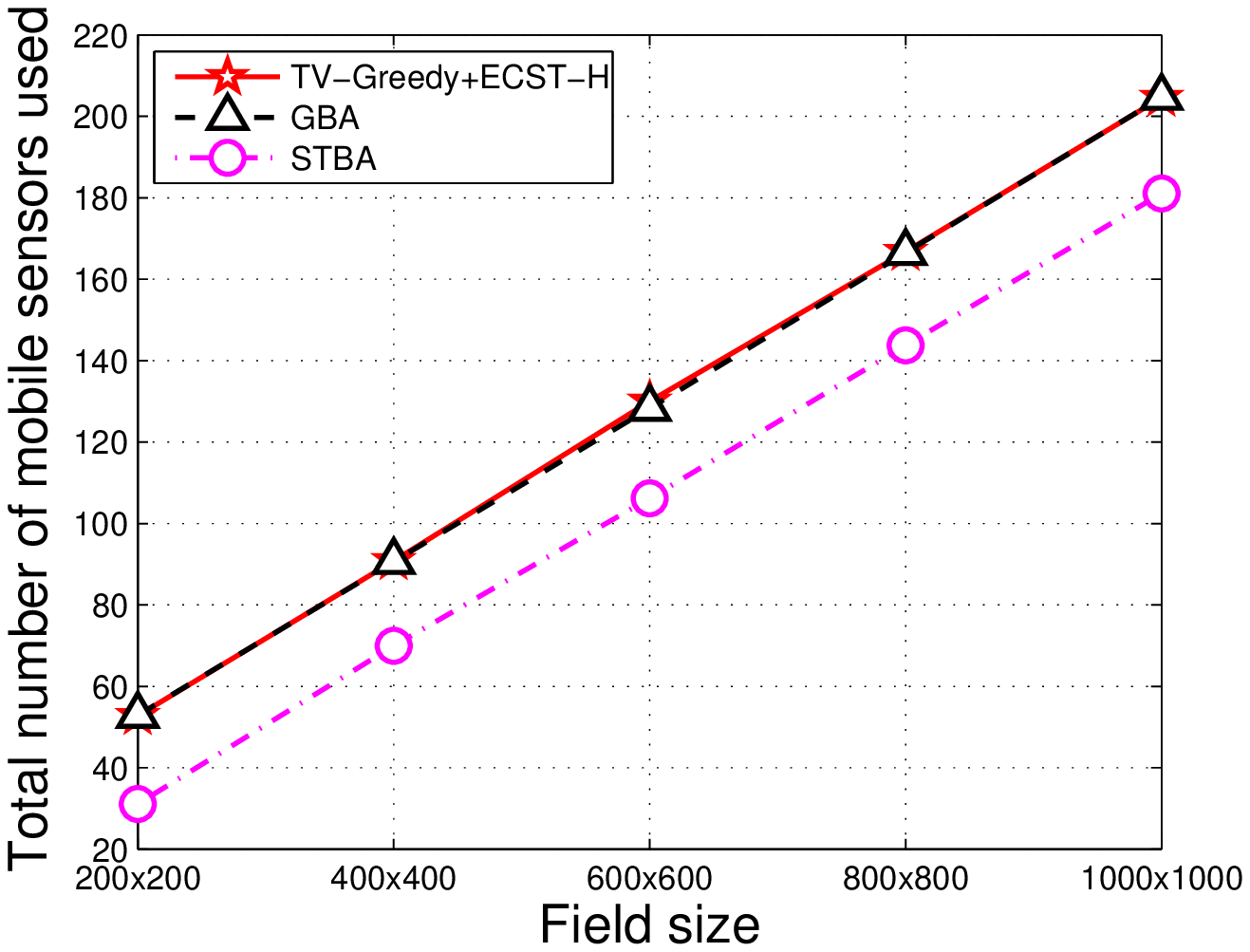}\label{Fig:size_sensor_full}}
\subfigure[]{\includegraphics[width=4.25cm]{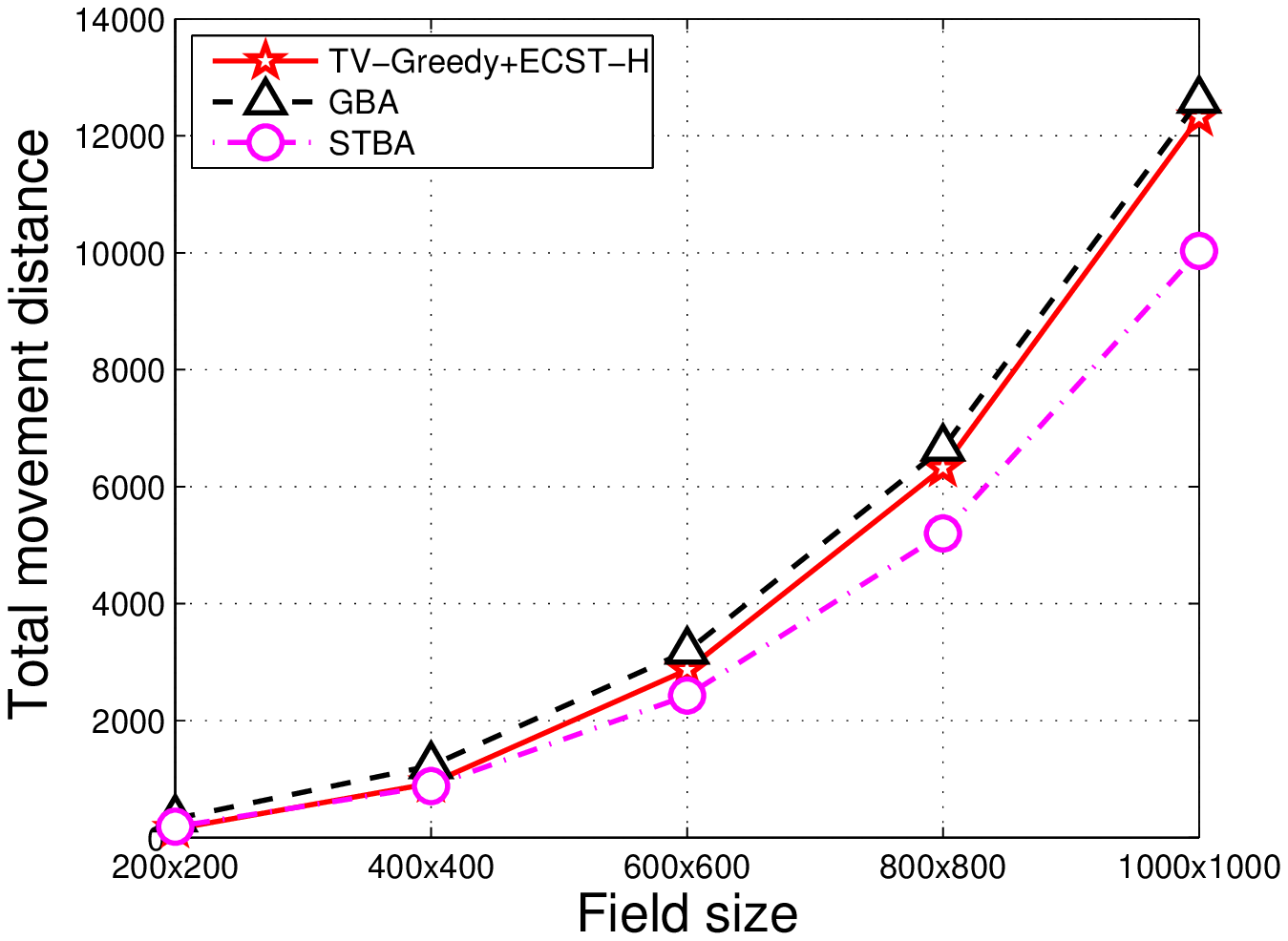}\label{Fig:size_distance_full}}
\caption{The total number of mobile sensors used and the total movement distance required in MWSNs whose field size ranges from $200$ $\times$ $200$ to $1000$ $\times$ $1000$. The required total number of mobile sensors used and the total movement distance are shown in (a) and (b), respectively.}
\end{figure}

Fig. \ref{Fig:size_sensor_full} and Fig.
\ref{Fig:size_distance_full} show the comparisons of the total
number of mobile sensors used and the total movement distance,
respectively, in MWSNs when the field size ranges from $200$
$\times$ $200$ to $1000$ $\times$ $1000$. In Fig.
\ref{Fig:size_sensor_full}, the larger the field
size, the higher the number of mobile sensors used by the
TV-Greedy+ECST-H, the GBA, and the STBA. This is because more mobile
sensors are required to maintain the network connectivity. In addition,
the STBA outperforms the TV-Greedy+ECST-H and the GBA because the
potential points generated by the STBA are as low as possible,
as explained for the results in Fig. \ref{Fig:target_num_sensor}. In
Fig. \ref{Fig:size_distance_full}, the larger the
field size, the longer the total movement distance of the mobile sensors
required by the TV-Greedy+ECST-H, the GBA, and the STBA. This is
because more mobile sensors are required to cover targets and form a
connected network in a larger sensing field. In addition, the STBA
has a lower total movement distance than the others because fewer
mobile sensors are required by the STBA.

\begin{figure}
\center
\subfigure[]{\includegraphics[width=4.25cm]{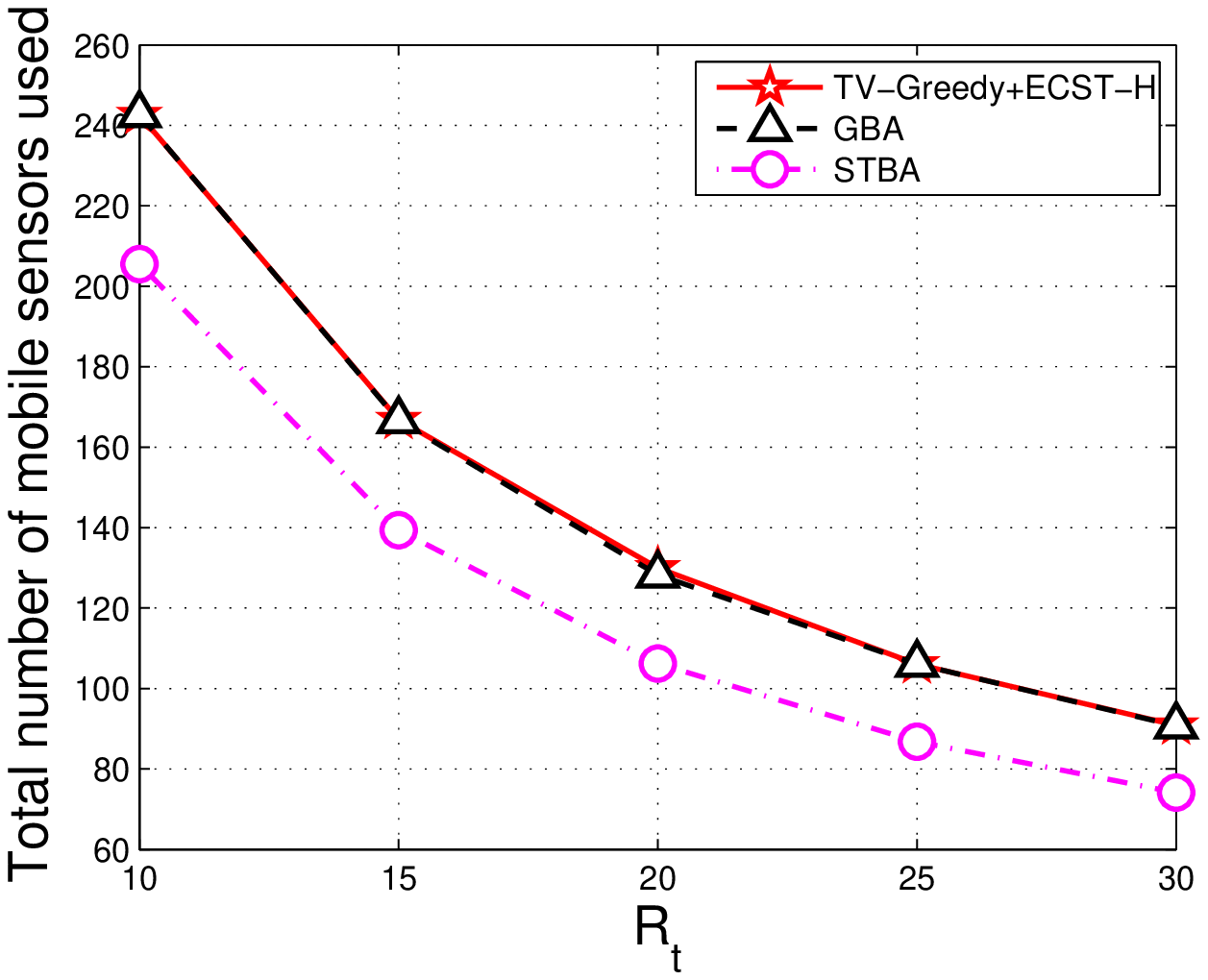}\label{Fig:transmission_sensor_full}}
\subfigure[]{\includegraphics[width=4.25cm]{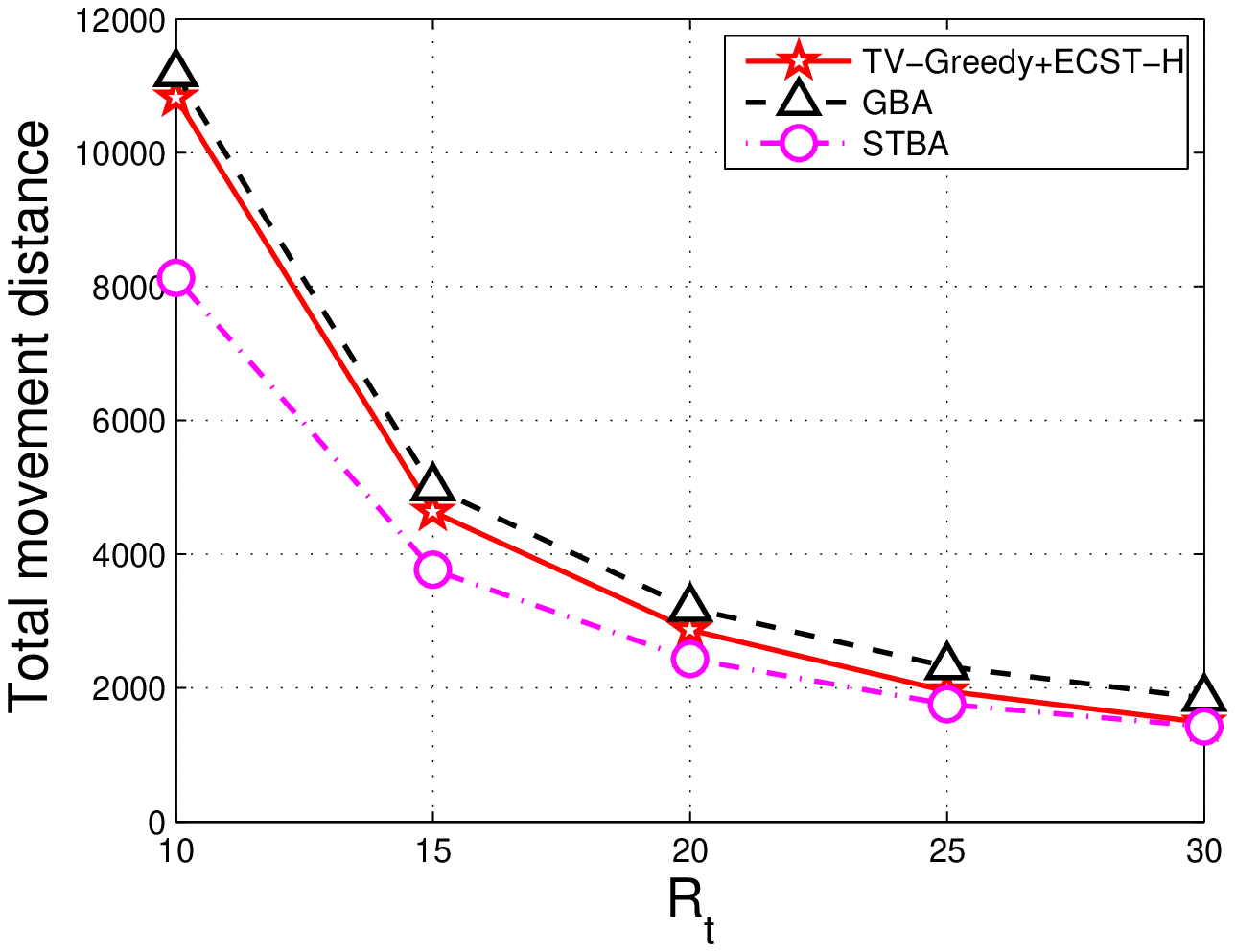}\label{Fig:transmission_distance_full}}
\caption{The total number of mobile sensors used and the total movement distance required in MWSNs whose mobile sensor transmission range ranges from $10$ to $30$. The required total number of mobile sensors used and the total movement distance are shown in (a) and (b), respectively.}
\end{figure}

Fig. \ref{Fig:transmission_sensor_full} and Fig.
\ref{Fig:transmission_distance_full} show the comparisons of the
total number of mobile sensors used and the total movement distance,
respectively, in MWSNs when the $R_t$ ranges from $10$ to $30$. In
Fig. \ref{Fig:transmission_sensor_full} and Fig.
\ref{Fig:transmission_distance_full}, the higher
the value of $R_t$, the lower the number of mobile sensors and the lower
the total movement distance required by the TV-Greedy+ECST-H, the
GBA, and the STBA. This is because fewer mobile sensors are required
to maintain network connectivity. In addition,
the STBA outperforms the TV-Greedy+ECST-H and the GBA in terms of
the number of mobile sensors used and the total movement distance
because the STBA generates as few potential points as
possible, as explained for the results in Fig.
\ref{Fig:target_num_sensor}.

\subsection{MWSNs in the RMWTCSCLMS
Problem}\label{section:sim:RMWTCSCLMS}

\begin{figure}
\center
\subfigure[]{\includegraphics[width=4.25cm]{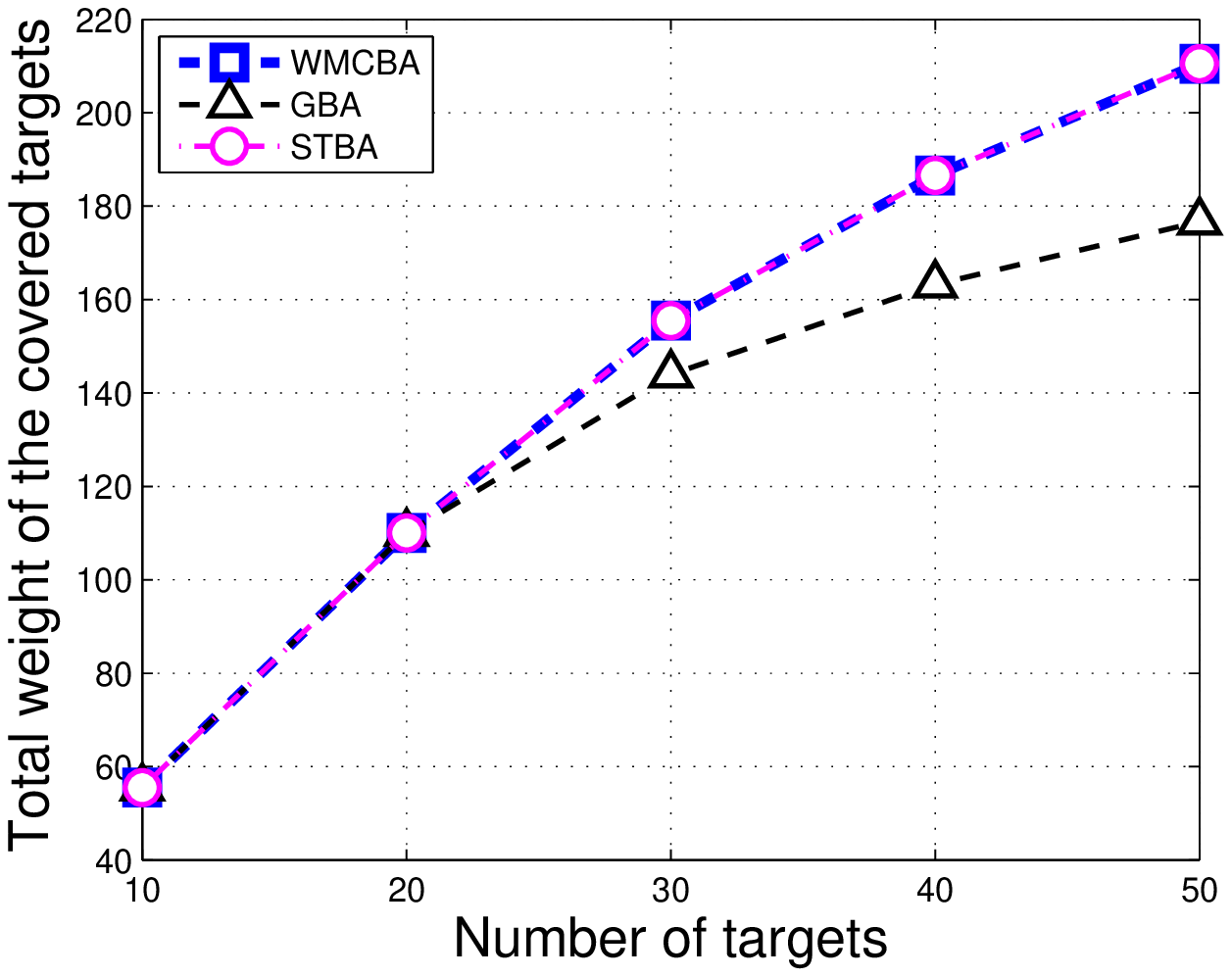}\label{Fig:target_total_weight_unlimitedTR}}
\subfigure[]{\includegraphics[width=4.25cm]{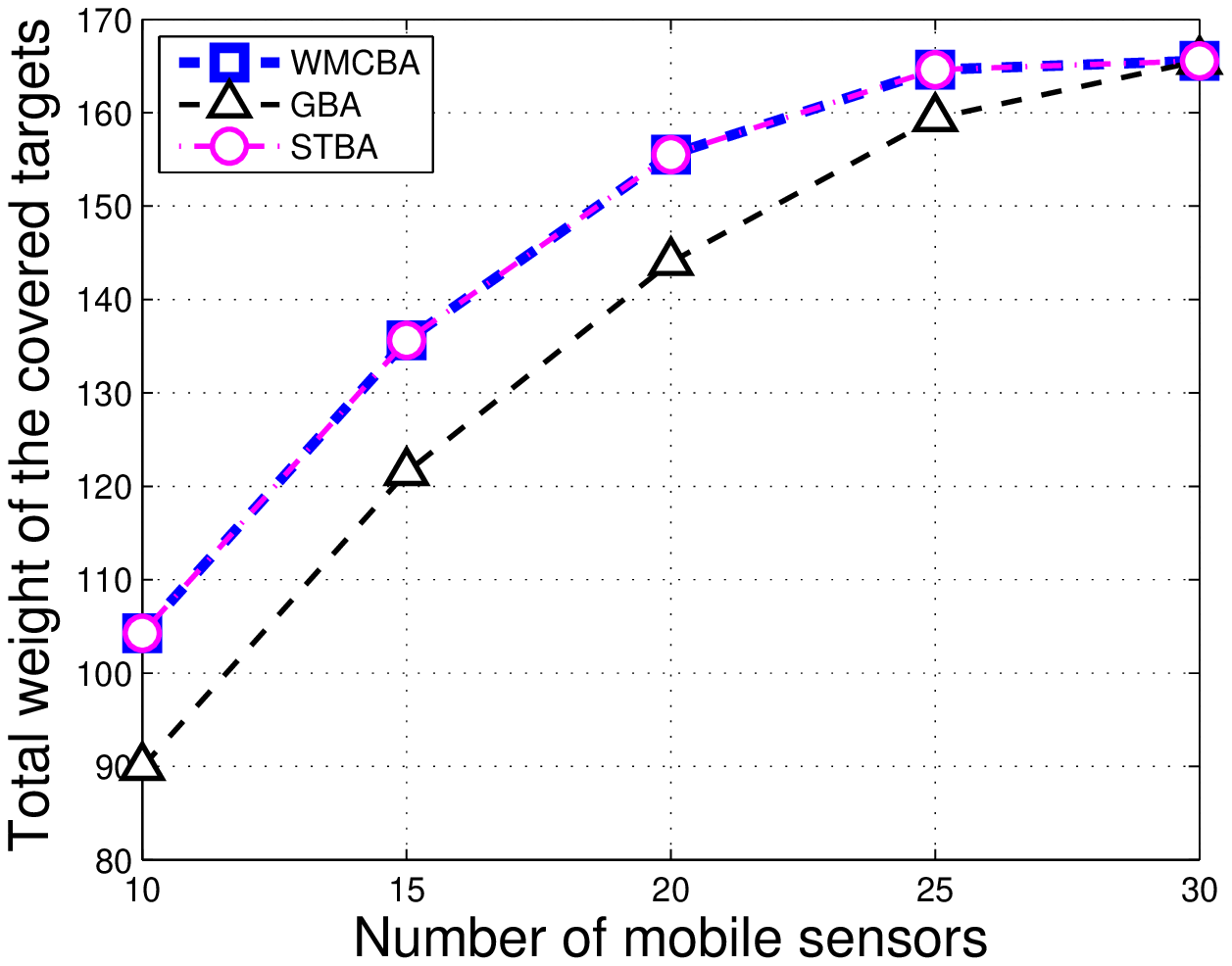}\label{Fig:sensor_total_weight_unlimitedTR}}\\
\subfigure[]{\includegraphics[width=4.25cm]{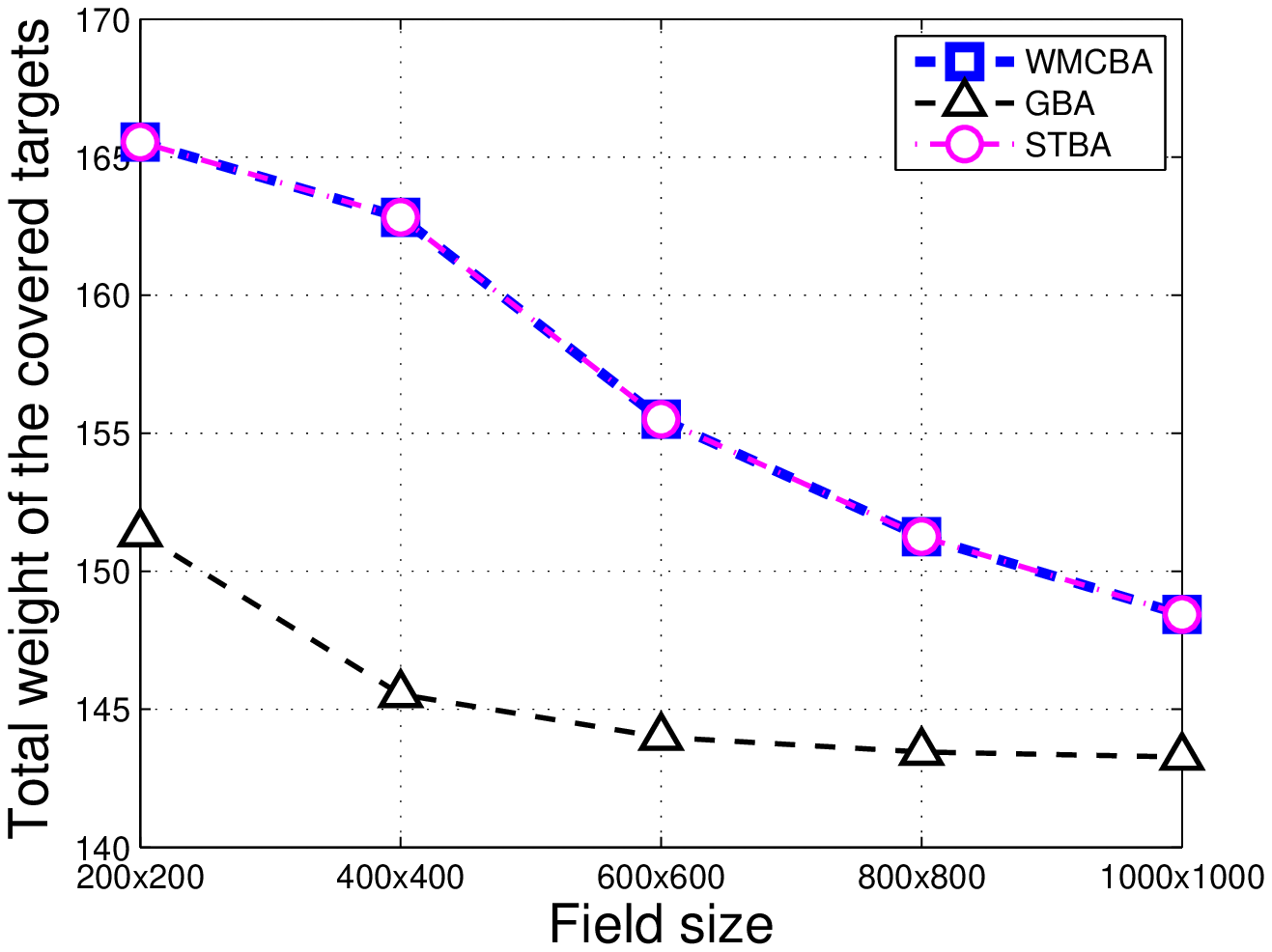}\label{Fig:size_total_weight_unlimitedTR}}\\
\caption{The total weight of the covered targets in MWSNs. The number of
targets ranging from $10$ to $50$, the number of mobile sensors
ranging from $10$ to $30$, and the field size ranging from $200$
$\times$ $200$ to $1000$ $\times$ $1000$ in MWSNs are shown in (a),
(b), and (c), respectively.}\label{Fig:unlimitedTR}
\end{figure}

In the MWSNs of the RMWTCSCLMS problem, unless otherwise stated, the
number of targets was set to $30$; and the number of mobile sensors was
set to $20$. Fig.
\ref{Fig:target_total_weight_unlimitedTR}, Fig.
\ref{Fig:sensor_total_weight_unlimitedTR}, and Fig.
\ref{Fig:size_total_weight_unlimitedTR} illustrate the total weight
of the covered targets in MWSNs with the number of targets ranging from
$10$ to $50$, in MWSNs with the number of mobile sensors ranging
from $10$ to $30$, and in MWSNs with the field size ranging from
$200$ $\times$ $200$ to $1000$ $\times$ $1000$, respectively. In
Fig. \ref{Fig:target_total_weight_unlimitedTR}, Fig.
\ref{Fig:sensor_total_weight_unlimitedTR}, and Fig.
\ref{Fig:size_total_weight_unlimitedTR}, the WMCBA
and the STBA have a higher total weight of the covered targets than the
GBA. This is because all possible sets of targets that can be
covered by any point in a sensing field are considered in the WMCBA
and the STBA, and thus, it has a high probability of selecting fewer
mobile sensors to cover the targets. Therefore, the remaining mobile
sensors can be used to cover other targets or maintain network
connectivity. In addition, the WMCBA and the STBA have the
same results. This is because for any instance of the RMWTCSCLMS
problem, the selection of covering targets in the STBA works in the same
greedy manner as in the WMCBA. In Fig.
\ref{Fig:target_total_weight_unlimitedTR}, the higher the number of
targets, the higher the total weight of the covered targets obtained by
the WMCBA, the GBA, and the STBA. This stems from the fact that more
targets can be covered by the mobile sensors. In Fig.
\ref{Fig:sensor_total_weight_unlimitedTR}, the higher the number of
mobile sensors, the higher the total weight of the covered targets
obtained by the WMCBA, the GBA, and the STBA because more mobile
sensors can be used to cover the targets. In Fig.
\ref{Fig:size_total_weight_unlimitedTR}, the larger the field size,
the lower the total weight of the covered targets obtained by the WMCBA,
the GBA, and the STBA. This is because fewer targets can be covered
by exactly one mobile sensor in a large sensing field, and thus,
fewer targets can be covered by $20$ mobile sensors.

\subsection{MWSNs in the MWTCSCLMS
Problem}\label{section:sim:MWTCSCLMS}

\begin{figure}
\center
\subfigure[]{\includegraphics[width=4.25cm]{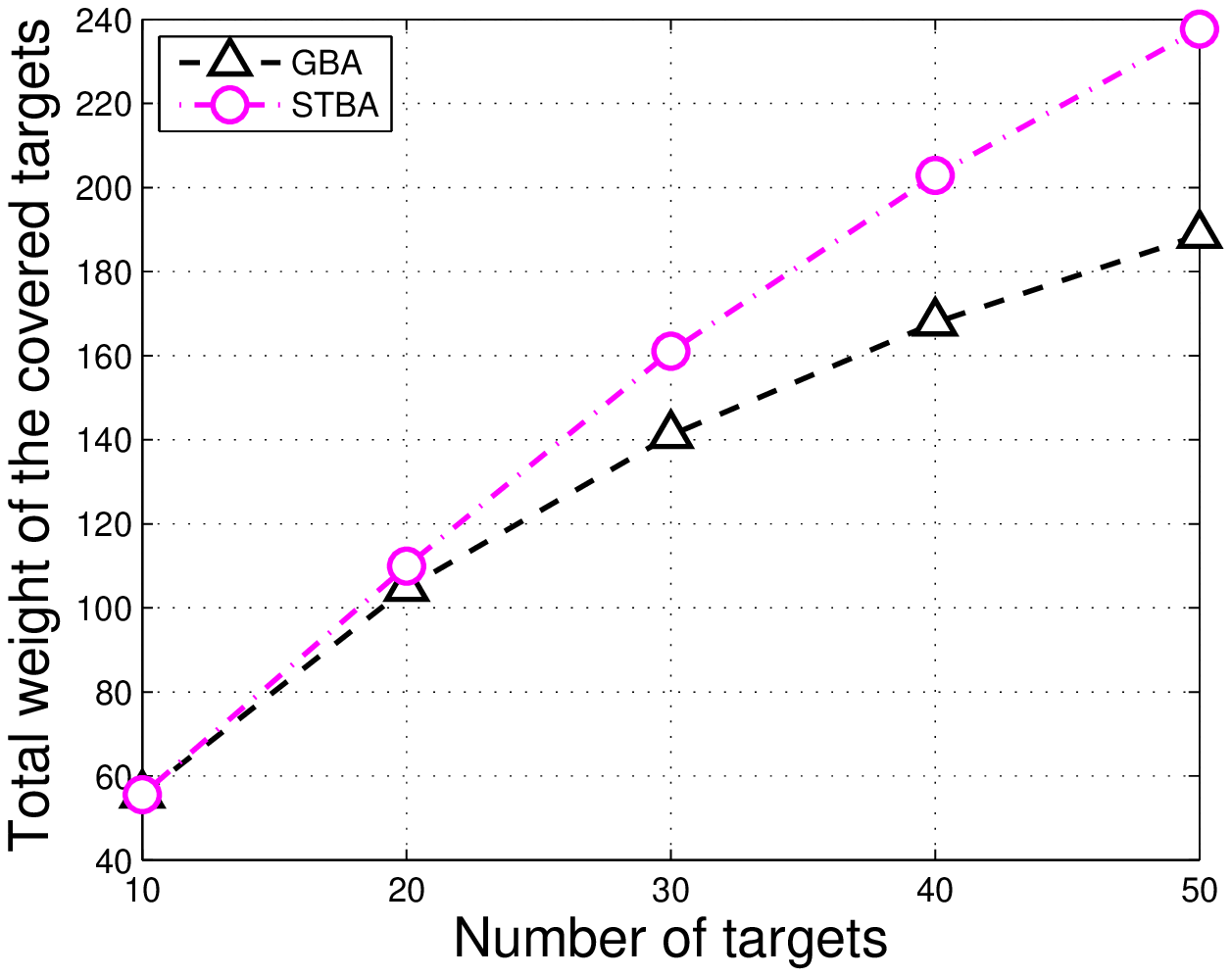}\label{Fig:target_total_weight}}
\subfigure[]{\includegraphics[width=4.25cm]{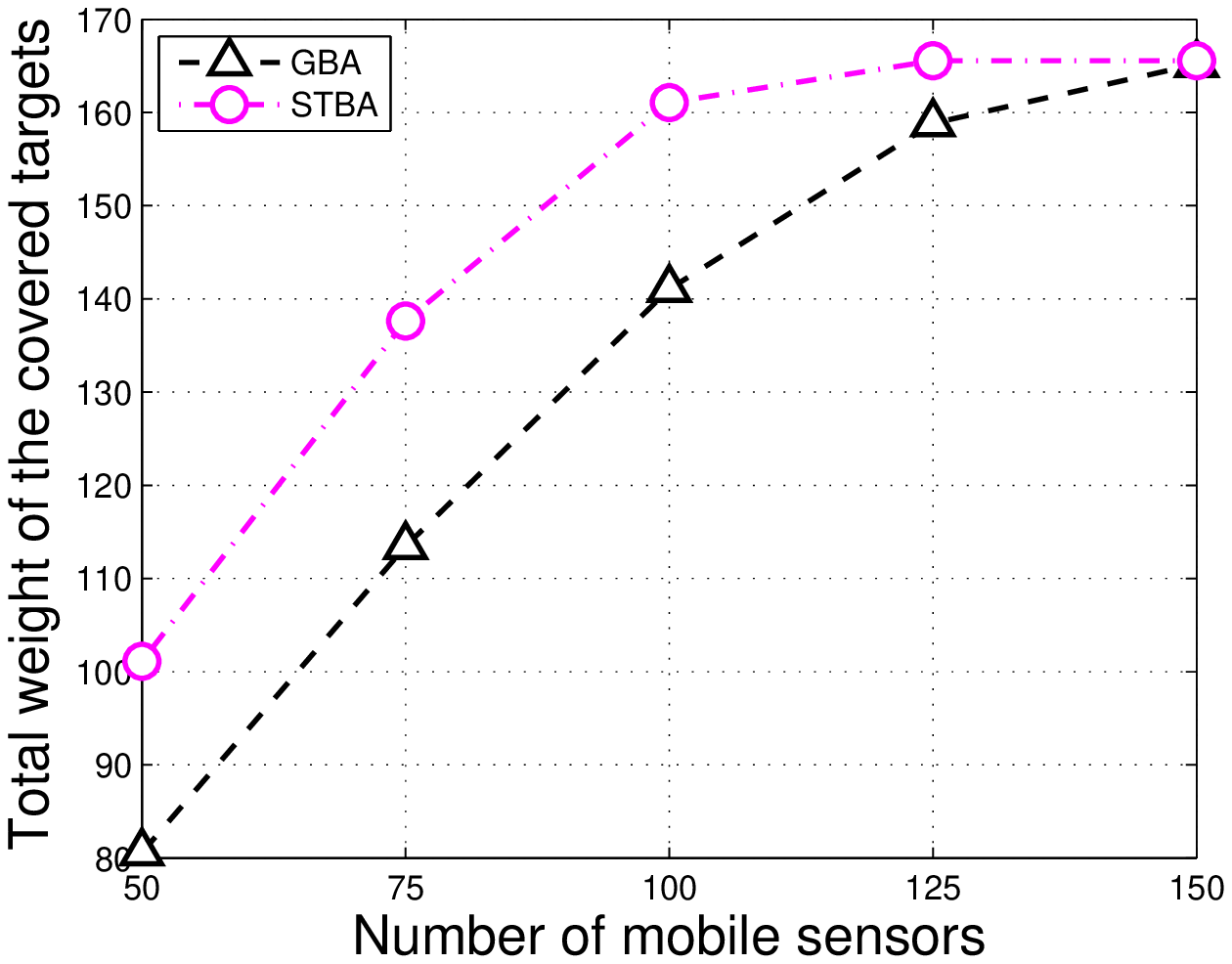}\label{Fig:sensor_total_weight}}\\
\subfigure[]{\includegraphics[width=4.25cm]{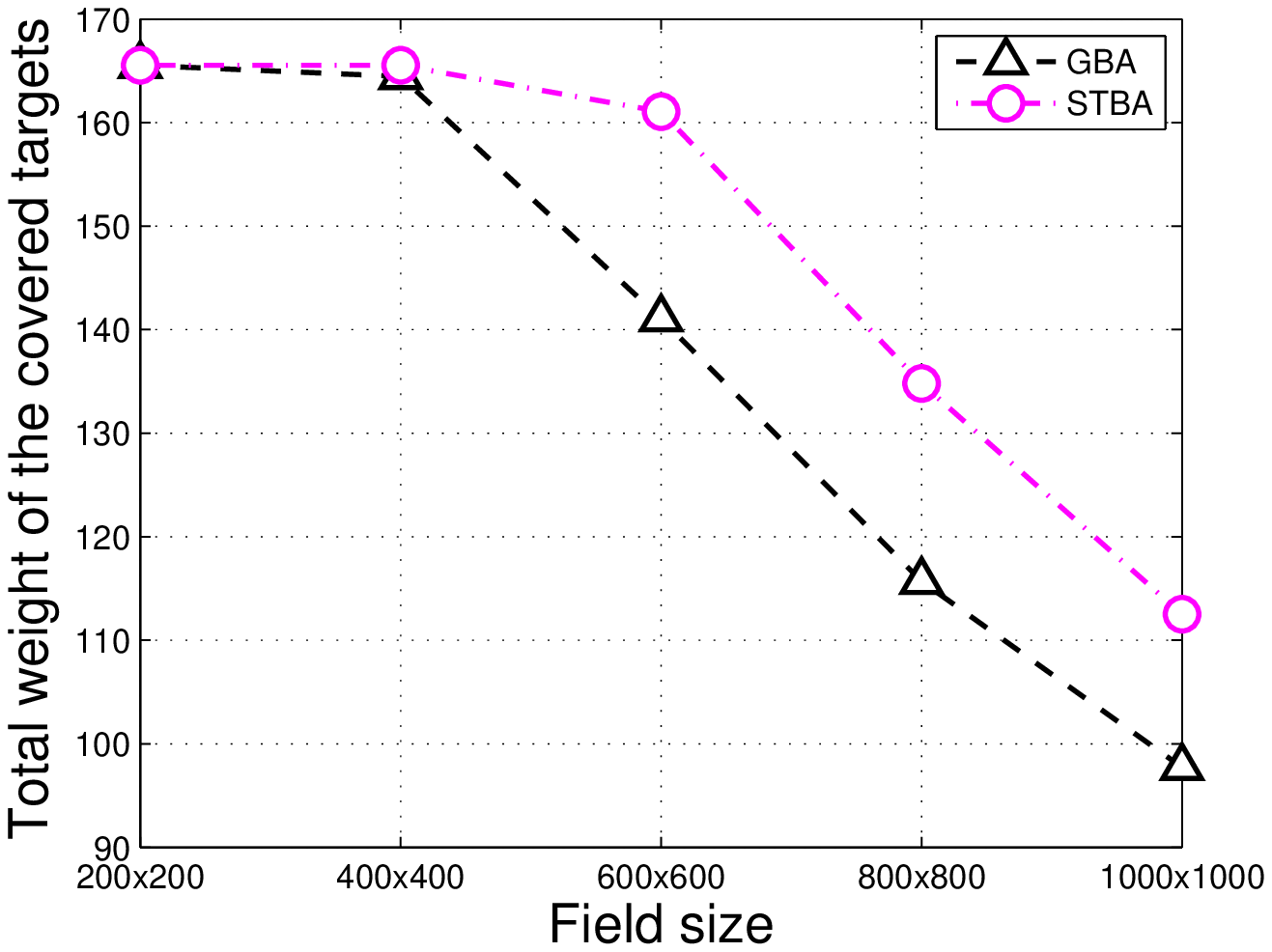}\label{Fig:size_total_weight}}
\subfigure[]{\includegraphics[width=4.25cm]{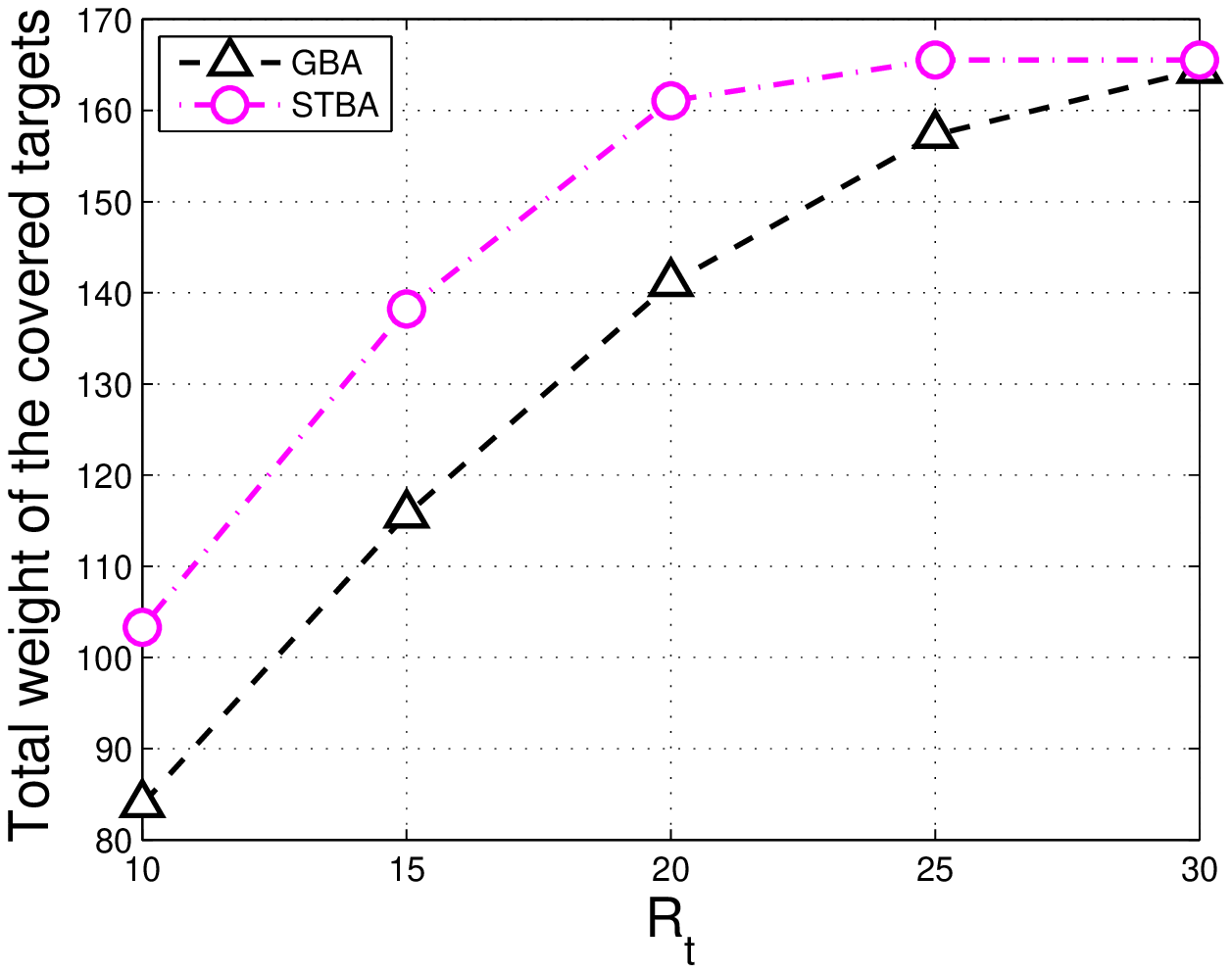}\label{Fig:transmission_total_weight}}\\
\caption{The total weight of the covered targets in MWSNs. The number of targets ranging from $10$ to $50$, the number of mobile sensors ranging from $50$ to $150$, the field size ranging from $200$ $\times$ $200$ to $1000$ $\times$ $1000$, and the $R_t$ ranging from $10$ to $30$ in MWSNs are shown in (a), (b), (c), and (d), respectively.}
\end{figure}

In the MWSNs of the MWTCSCLMS problem, unless otherwise stated, the
number of targets was set to $30$; and the number of mobile sensors was
set to $100$. Fig. \ref{Fig:target_total_weight}, Fig.
\ref{Fig:sensor_total_weight}, Fig. \ref{Fig:size_total_weight}, and
Fig. \ref{Fig:transmission_total_weight} show the total weight of
the covered targets in MWSNs with the number of targets ranging from
$10$ to $50$, in MWSNs with the number of mobile sensors ranging
from $50$ to $150$, in MWSNs with the field size ranging from $200$
$\times$ $200$ to $1000$ $\times$ $1000$, and in MWSNs with the
$R_t$ ranging from $10$ to $30$, respectively. In Fig.
\ref{Fig:target_total_weight}, Fig. \ref{Fig:sensor_total_weight},
Fig. \ref{Fig:size_total_weight}, and Fig.
\ref{Fig:transmission_total_weight}, the STBA has
a higher total weight of the covered targets than the GBA because more
targets can be covered by the STBA, as explained for the results in
Fig. \ref{Fig:unlimitedTR}. In addition, the
results of the GBA and the STBA in Fig.
\ref{Fig:target_total_weight}, Fig.
\ref{Fig:sensor_total_weight}, and Fig.
\ref{Fig:size_total_weight} are similar to those in Fig.
\ref{Fig:target_total_weight_unlimitedTR}, Fig.
\ref{Fig:sensor_total_weight_unlimitedTR}, and Fig.
\ref{Fig:size_total_weight_unlimitedTR}, respectively, as explained
for the results in Fig. \ref{Fig:unlimitedTR}, except for the results in Fig.
\ref{Fig:size_total_weight} with a small field size. In Fig.
\ref{Fig:size_total_weight}, when the field size is smaller than $600 \times 600$,
the GBA and the STBA have similar results. This is because almost all targets are covered by the mobile sensors
in the GBA and the STBA in these cases.
Moreover, in Fig.
\ref{Fig:transmission_total_weight}, the higher the $R_t$ value, the
higher the total weight of the covered targets obtained by the GBA and
the STBA. This is because fewer mobile sensors are used for network
connectivity, and thus, more mobile sensors can be used to cover
targets.

Fig. \ref{fig:same_instance_GBA} and Fig.
\ref{fig:same_instance_STBA} illustrate the deployment orders
generated by the GBA and the STBA, respectively, for the MWSN, in
which $100$ mobile sensors and $30$ targets were randomly generated
in a $600 \times 600$ sensing field, and $R_s$ and $R_t$ were set to
$20$. The total weight of the covered targets obtained by the GBA is
147, and that obtained by the STBA is 165.

\begin{figure}
\center \subfigure[]{\includegraphics[width=4.25cm]{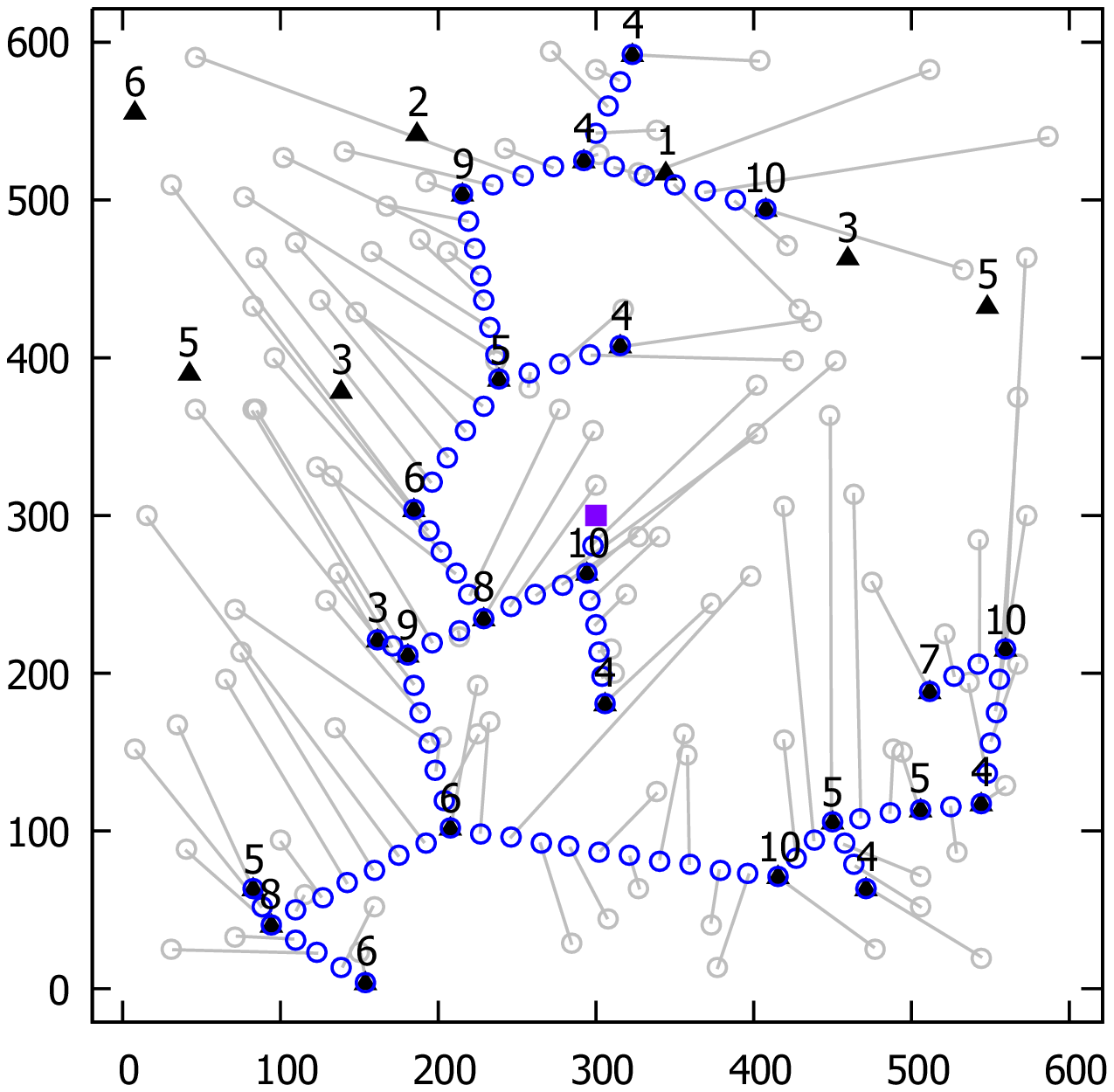}\label{fig:same_instance_GBA}}
\subfigure[]{\includegraphics[width=4.25cm]{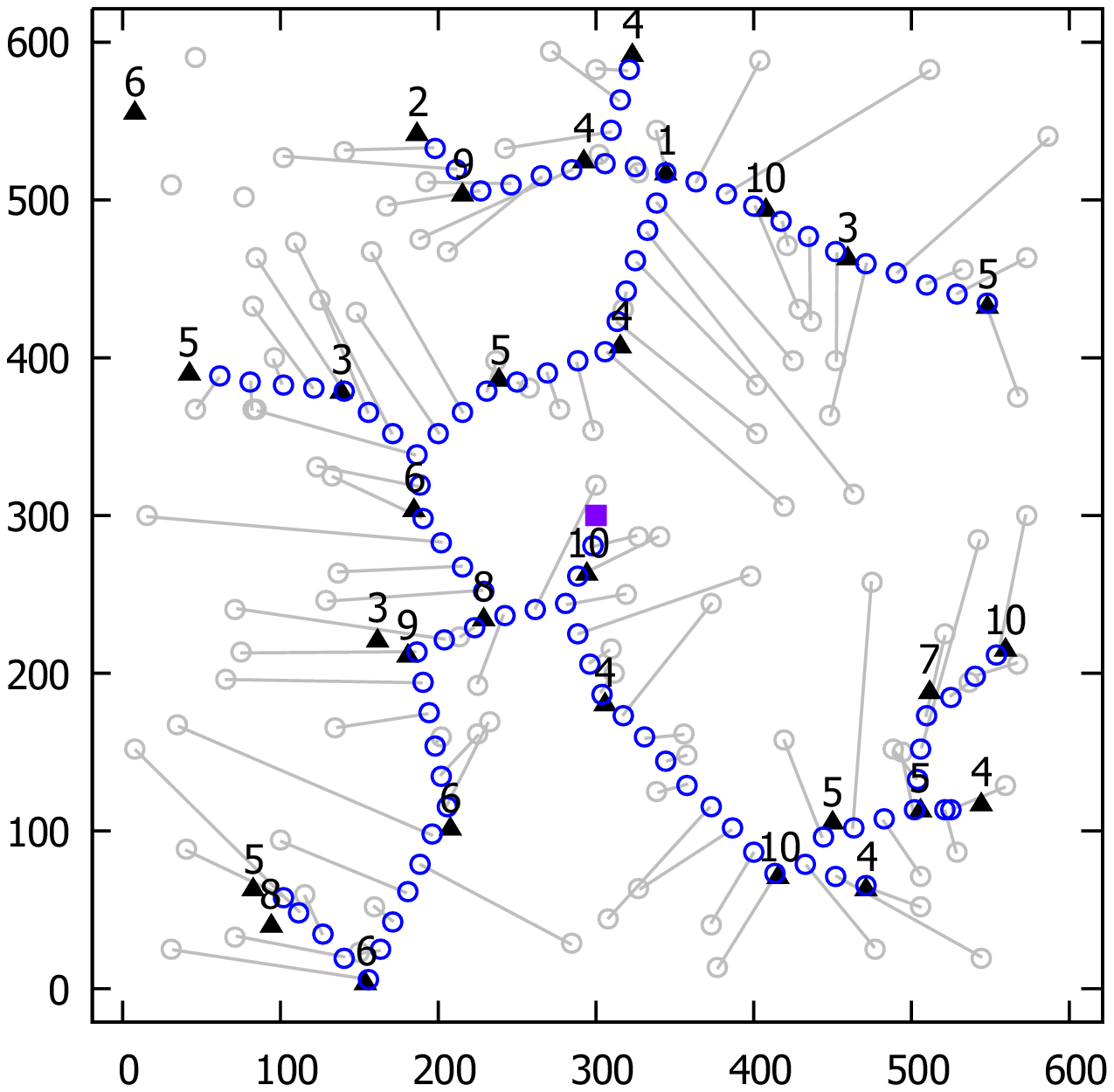}\label{fig:same_instance_STBA}}\\
\subfigure{\includegraphics[width=6cm]{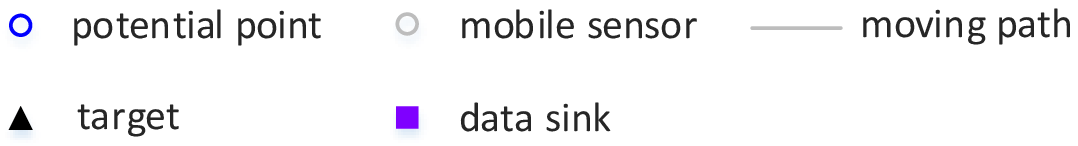}}
\caption{Deployment orders generated by the GBA and the STBA. The
results for the GBA and the STBA are shown in (a) and (b),
respectively.}
\end{figure}

\section{Conclusion}\label{section:Conclusion}
In this paper, the problem of scheduling limited mobile sensors to
appropriate locations to cover targets and form a connected network
such that the total weight of the covered targets is maximized, termed
the Maximum Weighted Target Coverage and Sensor Connectivity with
Limited Mobile Sensors (MWTCSCLMS) problem, was investigated. In
addition, a subproblem of the MWTCSCLMS problem, termed the
RMWTCSCLMS, was also investigated and analyzed. The RMWTCSCLMS
problem and the MWTCSCLMS problem were shown to be NP-hard here.
Moreover, an approximation algorithm, termed the
weighted-maximum-coverage-based algorithm (WMCBA), was proposed for
the RMWTCSCLMS problem. Based on the WMCBA, the
Steiner-tree-based algorithm (STBA) was therefore proposed for the
MWTCSCLMS problem. Theoretical analyses of the WMCBA and the
STBA were also provided.

In the simulation, three MWSN scenarios were considered, including
dense MWSNs, MWSNs in the RMWTCSCLMS problem, and MWSNs
in the MWTCSCLMS problem. In dense MWSNs, enough mobile sensors
were provided such that all targets could be fully covered and form a
connected network. The simulation results showed that the STBA had a
significantly lower total movement distance than the
TV-Greedy+ECST-H that is the best solution for the MSD problem. In
the MWSNs of the RMWTCSCLMS problem, simulation results showed that
the STBA was comparable to the WMCBA. In the MWSNs of the MWTCSCLMS
problem, the STBA outperformed the greedy-based algorithm (GBA)
proposed in the simulation section for the MWTCSCLMS problem.

%\section{Acknowledgements} This work was supported by the Ministry of Science and
%Technology under Grant MOST abc-xyz.

\normalem

\bibliographystyle{IEEEtran}
\bibliography{my}

\end{document}